\documentclass[review]{elsarticle}

\usepackage{etex}
\usepackage[utf8]{inputenc}
\usepackage{natbib}
\usepackage{amsmath,amssymb}
\usepackage{latexsym}
\usepackage{booktabs}
\usepackage{ragged2e}
\usepackage{geometry}
\usepackage{stmaryrd}
\usepackage{graphicx}
\usepackage{float}
\usepackage{endnotes}
\usepackage{blkarray}
\usepackage{bbm}
\usepackage{indentfirst}
\usepackage{latexsym}
\usepackage{booktabs}
\usepackage{geometry}
\usepackage{amsthm}
\usepackage{subfig}
\usepackage{color}
\usepackage[export]{adjustbox}
\usepackage{listings}
\usepackage{mathtools}
\usepackage[ruled,linesnumbered]{algorithm2e}

\usepackage{epstopdf}
\usepackage{multirow}
\usepackage{paralist}
\usepackage{setspace}
\usepackage{morefloats}
\usepackage[]{algorithm2e}
\usepackage{mdwlist}
\usepackage{booktabs} 
\usepackage{amsfonts}
\newcommand{\be}{\begin{equation}}
\newcommand{\ee}{\end{equation}}

\newtheorem{thm}{Theorem}

\newtheorem{cor}[thm]{Corollary}

\newtheorem{defn}[thm]{Definition}

\newtheorem{prop}[thm]{Proposition}

\newcommand{\RNum}[1]{\uppercase\expandafter{\romannumeral #1\relax}}
\usepackage{hyperref}
\makeatletter
\def\tagform@#1{\maketag@@@{\bfseries(\ignorespaces#1\unskip\@@italiccorr)}}
\renewcommand{\eqref}[1]{\textup{{\normalfont(\ref{#1}}\normalfont)}}
\makeatother
\newcommand{\bigzero}{\mbox{\normalfont\Large\bfseries 0}}

\usepackage{thmtools}
\usepackage{lineno}
\usepackage{geometry}
\usepackage{color}
\usepackage[export]{adjustbox}
\usepackage{listings}
\newcommand{\ra}[1]{\overrightarrow{#1}}
\newcommand{\la}[1]{\overleftarrow{#1}}

\usepackage{algorithmic}
\usepackage{tikz}
\usepackage{caption}

\usepackage{blkarray}
\usepackage{amsmath}

\journal{}

\begin{document}

\begin{frontmatter}


\title{Modeling the spread of infectious disease in urban areas with travel contagion}

\author{Xinwu Qian, Satish Ukkusuri}

\address{Lyles School of Civil Engineering, Purdue University}

\begin{abstract}
Urban mass transportation system satisfies the essential mobility needs of the large-scale urban population, but it also creates an ideal environment that favors the spread of infectious diseases, leading to significant risk exposure to the massive urban population. In this study, we develop the mathematical model to understand the coupling between the spreading dynamics of infectious diseases and the mobility dynamics through urban transportation systems. We first describe the mobility dynamics of the urban population as the process of leaving from home, traveling to and from the activity locations, and engaging in activities. We then embed the susceptible-exposed-infectious-recovered (SEIR) process over the mobility dynamics and develops the spatial SEIR model with travel contagion (Trans-SEIR), which explicitly accounts for contagions both during travel and during daily activities. We investigate the theoretical properties of the proposed model and show how activity contagion and travel contagion contribute to the average number of secondary infections. In the numerical experiments, we explore how the urban transportation system may alter the fundamental dynamics of the infectious disease, change the number of secondary infections, promote the synchronization of the disease across the city, and affect the peak of the disease outbreaks. The Trans-SEIR model is further applied to the understand the disease dynamics during the COVID-19 outbreak in New York City, where we show how the activity and travel contagion may be distributed and how effective travel control can be implemented with only limited resources. The Trans-SEIR model along with the findings in our study may have significant contributions to improving our understanding of the coupling between urban transportation and disease dynamics, the development of quarantine and control measures of disease system, and promoting the idea of disease-resilient urban transportation networks. 
\end{abstract}


\end{frontmatter}



\section{Introduction}
As of 2018, there were over 4.2 billion people living in urban areas and the number is estimated to reach 6.3 billion (70\% of the total world population) in 2050~\cite{wordlpopu18}. The rapid growth of the urban population gives rise to the prosperity of the urban economy with much more travel and intensive daily activities. A notable example is the development of urban mass transit systems. As the only affordable mobility solution to the general urban public, it has grown to enormous scale in large cities with buses and metros serving over 53 billion passengers worldwide as of 2017~\cite{urban_trans_2018}. Behind the giant number, however, is an environment that favors the spread of infectious diseases which constitutes a significant health risk to the entire urban community. On the one hand, the urban population is making more trips and traveling a longer distance, with the average annual person-mile traveled in the United States increased by 169\% from 1969 to 2009~\cite{chambers2015passenger}. On the other hand, people are spending more time in the transportation system especially in urban areas. For New York City, the average commuting to work time exceeds 40 minutes for 69\% of the city's neighborhoods, and 59\% of New York commuters use mass transit as the tool for commuting. To meet the growing mobility needs with limited mobility resources, the mass transit systems are developed to carry as many passengers, leading to enclosed compartments with high population density, close proximity, and long duration. This represents sufficient exposure duration and a close-enough distance that allows for the pathogens to migrate among unprotected daily commuters. A direct consequence is the recent outbreaks of the COVID-19 in major cities worldwide~\cite{dong2020interactive} with over 4 million confirmed cases and 270 thousand deaths as of May 2020. Despite the significant risks associated with the urban transportation system and the outrageous consequences caused by numerous infectious diseases, the role played by the urban transportation system during the outbreaks of infectious diseases in urban areas remains largely unexplored. This motivates us to make the initial attempt to model the contagion process during travel and further explore possible mitigation measures that may help to curb the spread of infectious diseases. 

There are two main approaches in the literature for characterizing the dynamics of infectious diseases. The first approach comprises of the compartment model, where the initial SIR model (also known as Kermack-McKendrick model) divides the population into compartments of susceptible, infected, and recovered, and non-linear ordinary differential equations (ODE) are used to model the dynamics among the compartments~\cite{kermack1927contribution}. Based on the SIR model, a variety of models have been developed to account for more realistic disease nature, including models which consider the incubation period~\cite{li1995global}, the vertical transmission~\cite{smith2001global}, the age structure~\cite{castillo1989epidemiological}, and the vaccine strategy~\cite{shulgin1998pulse,d2005pulse}. And extensive efforts were made to understand the property of the model and analyze the local and global stability of the non-linear ODEs, where surveys of related works can be found in~\cite{hethcote1981periodicity,liu1987dynamical}. One major criticism for the compartment model is the oversimplified assumptions, where the population in each compartment is assumed to be fully mixed and therefore each individual has the same behavior. Therefore, the model may fail to represent the complex mobility and contact patterns for many real-world diseases. To address this shortcoming, the second approach models the disease propagation at the individual level over large-scale networks, where each individual is represented by a node and their contact structure is captured by the set of edges~\cite{newman2002spread,meyers2007contact}. The approach was developed based on the method of bond percolation model, and the generation function was used for deriving the important attributes of a certain contact network, including the average degree and excessive degree, which later used to calculate the size of the disease outbreaks. Danon et al.~\cite{danon2011networks} conducted a comprehensive review of the works in this approach. The contact network model helps to capture diverse interactions among individuals with given distribution and contributes to the understanding between disease dynamics and the network topology. However, it is not applicable to understand the spread of infectious disease in densely populated urban areas: it is computationally intractable to construct the individual contact network for all urban populations and it is also impossible to obtain the necessary input for the contact network as evaluating the individual's contact pattern will be very expensive. 

In light of the existing efforts and issues,  to model the spread of infectious diseases for the urban areas, the compartment model is still the ideal choice due to its well-explored mathematical properties and the scalability when it comes to model the mass urban population. However, the simplistic assumption of homogeneous population mixing needs to be corrected as the population dynamics are highly heterogeneous in large cities. In particular, since human mobility is the driving factor for the spread of infectious diseases, it is necessary for the compartment model to account for the mobility dynamics around the city. And the needs of mobility dynamics assert two additional challenges for developing appropriate compartment models to represent the urban disease dynamics. First, the diverse land use and activity patterns in urban areas suggest the necessity to incorporate the spatial-varying mobility patterns and activity patterns into each of the original disease-related compartments, so that the rate of transmission and the amount of population should be differentiated at individual subareas of a city. On the other hand, while the spatial movements bring people to their activity locations, it also results in massive contact and contagion during travel which is negligible given the existence of urban mass transit systems. This together with the spatial heterogeneity require additional consideration to explicitly capture the contact and the disease transmission during travel as separate compartments. To address these challenges, we propose the spatial SEIR model with travel contagion (Trans-SEIR) to characterize the spread of infectious diseases in urban areas. The proposed model overcomes the first challenge by developing an urban mobility model to represent the mobility dynamics of leaving home, traveling to-and-from activity locations and activity engagements. Then the SEIR process can then be overlaid with the mobility model to reflect the spatial heterogeneity of the disease dynamics. In addition, the Trans-SEIR model considers two types of contagion, namely the activity contagion and travel contagion, to explicitly separate the infections due to urban travel activities and the infections that arise from daily activities such as work and entertainment. Based on the Trans-SEIR model, we further explore the implementation of entrance control in the urban transportation system and develop the optimal resource allocation problem which determined the distribution of limited medical resources and manpower over the travel segments in urban areas to minimize the risk from the infectious diseases. To this end, the proposed Trans-SEIR model represents a realistic modeling framework to understand the complete trajectories of disease outbreaks in urban areas, and serves as the essential modeling components for devising optimal control strategies and policies, which have significant implications to both preventative prevention as well as disease mitigation from the urban transportation system perspective.

The rest of the study is organized as follows. In the next section, we briefly introduce the background of epidemic modeling with compartment models and establish the basics for modeling the mobility dynamics. In Section 3, we present the mathematical formulations for the Trans-SEIR model and analyze the theoretical properties of the model, followed by the optimal entrance control model for urban transportation systems in Section 4. In Section 5, we present the numerical experiments of both toy networks and the real-world case study for the COVID-19 outbreak in New York City (NYC). Finally, we conclude our study with major findings and insights in Section 6. 

\section{Modeling preliminaries}

\subsection{Notation}
We summarize the list of variables used in this section as follows:
\begin{table}[htbp!]
	\setstretch{1.2}
	\centering
	\caption{Table of notation}
	\label{notation}
	\begin{tabular}{p{3cm}p{13cm}}
		\hline
		Notation & Description \\ \hline
		&\textbf{Variables}\\ \hline
		$S$ & Susceptible population. \\
		$E$ & Exposed (latent) population.\\
		$I$ & Infected population. \\
		$R$ & The population who recovered from disease and got immunity.\\
		$N_i^p$ & Total amount of visitors who are current present at zone $i$. \\
        $N_i^r$ & Resident population at zone $i$. \\ 
		$N_{ij}$ & The amount of people who are residents of zone $i$ and currently present at zone $j$. \\
		\hline
		&\textbf{Fixed parameters}\\ \hline
        $\beta$ & Contagion rate between $S$ and $I$. \\
        $1/\sigma$ & Length of latent period for population $E$.\\
        $1/\gamma$ & Length of infectious period for population $I$. \\
        $\mu$ & Death and birth rate. \\
        $P$ & Total number of zonees in the area. \\
        $\alpha_i$ & Arrival (departure) rate of external population for zone $i$. \\ 
		$g_i$ & Total departure rate of zone $i$.\\
		$m_{ij}$ & The rate of movement from zone $i$ to zone $j$, where $\sum_j m_{ij}=1$. \\
        $r_{ij}$ & The rate of return from zone $j$ to zone $i$\\
        $d^{M}$ & Control rate of travel mode $M$. \\
        \hline
	\end{tabular}
\end{table}

\subsection{Epidemic modeling}
The well-known compartment model for capturing the dynamics of infectious diseases was proposed by W.O.Kermack and A.G.McKendrick~\cite{kermack1927contribution}, where they consider that the population may experience three states over time: 
\begin{itemize}
\item Susceptible class or $S(t)$ is used to represent the number of individuals not yet infected with the disease at time t, or those susceptible to the disease.
\item Infected class or $I(t)$ denotes the number of individuals who have been infected with the disease and are capable of spreading the disease to those in the susceptible category.
\item Removed class or $R(t)$ is the compartment used for those individuals who have been infected and then removed from the disease, either due to immunization or due to death. Those in this category are not able to be infected again or to transmit the infection to others.
\end{itemize}

And the compartment model has several key assumptions: (1) each individual in the population has an equal probability of contracting the disease with a rate of $\beta$, (2) the population leaving the susceptible class is equal to the number of people entering the infected class, (3) people recovered from the disease with a mean recovery of $1/\gamma$ gain permanent immunity to the disease, and (4) the death rate is the same as the birth rate so that the total population is fixed. 

In reality, people who are infected by certain diseases may not present any symptoms until the end of the incubation period, and it is important to take this latent period into consideration for more accurate representation of disease dynamics. Consequently, the SEIR model was introduced with an additional compartment which is known as the latent class (E(t)) \cite{lloyd1996spatial}. The population of E are considered as exposed but not infectious, and will proceed into the infectious state with an average length of latent period of $\frac{1}{\sigma}$. Such disease dynamics can be mathematically represented as:
\begin{eqnarray}
&&\frac{dS}{dt}=-\beta SI+\mu (N-S) \\ \nonumber
\\ 
&&\frac{dE}{dt}=\beta SI-\mu E-\sigma E\\ \nonumber
\\ 
&&\frac{dI}{dt}=\sigma E-\gamma I-\mu I\\ \nonumber
\\
&&\frac{dR}{dt}=\gamma I-\mu R 
\end{eqnarray}

\subsection{Mobility model}
\label{sec:mobility}
The spread of infectious disease is closely interacted with the mobility pattern of urban population. Before discussing the mathematical model for capturing disease dynamics, we first present the mobility model that is followed by urban population. 

The mobility model used in this study is adapted from the intra-city mobility model proposed by Sattenspiel and Dietz~\cite{sattenspiel1995structured}. Different from the previous study, here we explicitly captures the mobility dynamics of population during travel which is especially important for modeling the mobility dynamics in densely populated urban areas. Specifically, we divide the urban area into a collection of $P$ zones and the population are further classified into two groups: residents and visitors. Residents of zone $i$ will be in one of the four states: staying at the resident location, travelling to the activity location, engaging in activities, and travelling back to the home location. Let $\ra{N}_{ij}$ be the population who are residents of zone $i$ and travel from $i$ to $j$, $N_{ij}$ be the population of residents of zone $i$ who are currently at zone $j$, $\la{N}_{ij}$ be the population who travel from $j$ back to $i$ and $N_i$ be the population who remain at the resident location at $i$. We have the following two equations to capture the total residents and visitors for a given zone:
\begin{equation}
N_i^r=\sum_{j=1}^{P} N_{ij}+{N}_{\ra{ij}}+{N}_{\la{ij}},\forall\,i
\label{eq:res}
\end{equation}
\begin{equation}
N_i^p=\sum_{j=1}^P N_{ji},\forall\,i
\label{eq:vis}
\end{equation} 
Equation~\ref{eq:res} suggests that the amount of residents at zone $i$ can be calculated as the summation of population who are residents of zone $i$ and currently at zone $j$. Similarly, equation~\ref{eq:vis} states that visitor population at zone $i$ consist of the residents of zone $i$ who remain in $i$, as well as the population who reach zone $i$ from other zones.

For people in different groups, we assume a constant rate of death $\mu$ and an equal rate of birth. New births of zone $i$ are assumed to join the residents of the same zone ($N_i$). For residents at zone $i$, a constant departure rate is assumed ($g_i$) and the departed population may visit zone $j$ for various daily activities (e.g., work or entertainment) with the split ratio of $m_{ij}$ and $\sum_j m_{ij}=1$. The departed population then join the population in travel and will arrive at the activity location with a rate of $\alpha_{\ra{ij}}$. The population at activity location will return to their home at the rate of $r_{ij}$, who again join the population in travel and eventually arrive at home with the rate of $\alpha_{\la{ij}}$. With the above description, we have the following system of ordinary differential equations (ODE) for computing the mobility dynamics within a urban area:
\begin{equation}
\frac{d{N}_{\ra{ij}}}{dt}=g_im_{ij}N_{i}-\alpha_{\ra{ij}} {N}_{\ra{ij}}-\mu{N}_{\ra{ij}}  
\label{eq:bt1}
\end{equation}
\begin{equation}
\frac{d{N}_{\la{ij}}}{dt}=r_{ij}N_{ij}-\alpha_{\la{ij}} {N}_{\la{ij}}-\mu{N}_{\la{ij}}  
\label{eq:bt2}
\end{equation}

\begin{equation}
\frac{dN_{ij}}{dt}=-r_{ij}N_{ij}+\alpha_{\ra{ij}} {N}_{\ra{ij}} -\mu N_{ij}
\label{eq:bt3}
\end{equation}
\begin{equation}
\frac{dN_{i}}{dt}=\sum_{j=1}^P \alpha_{\la{ij}} {N}_{\la{ij}}-g_iN_{i} -\mu N_i+\mu N_i^r
\label{eq:bt4}
\end{equation}

With the above system of equations, we can calculate the rates of change for resident and visitor population as:
\begin{equation}
\begin{aligned}
\frac{{N}_i^r}{dt}&=\sum_{j=1}^P \frac{{N}_{ij}}{dt}+\frac{{N}_{\la{ij}}}{dt}+\frac{{N}_{\ra{ij}}}{dt}
\end{aligned}
\label{eq:res}
\end{equation} 

\begin{equation}
\begin{aligned}
\frac{{N}_i^p}{dt} &=\sum_{j=1}^P \frac{{N}_{ji}}{dt}
\end{aligned}
\label{eq:vis}
\end{equation}
The described system of equations will reach the equilibrium if equations~\ref{eq:bt1} to~\ref{eq:bt4} equal zero, e.g., there is no change of population for each compartment. Consequently, for the equilibrium solution of the mobility dynamics, we arrive at the following proposition:
\begin{prop}
	The system described by equations~\ref{eq:bt1} to~\ref{eq:bt4} is linear and has a unique equilibrium solution, and the solution is globally asymptotically stable. In particular, at equilibrium, we have
	\begin{equation}
	N_{i}^*+\sum_{j=1}^P K_{ij}N_{ij}^*=N_i^r
    \label{eq:equilibrium_ii}
	\end{equation}
	\begin{equation}
	N_{i}^*=\frac{1}{1+\sum_{j=1}^P K_{ij}}N_i^r
    \label{eq:equilibrium_ij}
	\end{equation}
    where $K_{ij}=g_im_{ij}(\frac{1}{r_{ij}}+\frac{1}{\alpha_{\la{ij}}}+\frac{1}{\alpha_{\ra{ij}}})$.
\end{prop}
\begin{proof}
	Since the system is linear, and the matrix constitutes a strongly connected graph. The equilibrium solution of the system can be calculated by setting equation~\ref{eq:bt1} to~\ref{eq:vis} to zero. 
    We have:
    \begin{equation}
    {N}_{\ra{ij}}^*=\frac{g_im_{ij}}{\alpha_{\ra{ij}}+\mu}N_{ii}^*, N_{ij}^*=\frac{g_im_{ij}\alpha_{\ra{ij}}}{(r_{ij}-\mu)(\alpha_{\ra{ij}}+\mu)}N_{ii}^*, ,{N}_{\la{ij}}^*=\frac{g_im_{ij}\alpha_{\ra{ij}}r_{ij}}{(r_{ij}-\mu)(\alpha_{\ra{ij}}+\mu)(\alpha_{\la{ij}}+\mu)}N_{ii}^*
    \label{eq:ijii}
    \end{equation}
    Let $K_{ij}=\frac{g_im_{ij}}{\alpha_{\ra{ij}}+\mu}+\frac{g_im_{ij}\alpha_{\ra{ij}}}{(r_{ij}-\mu)(\alpha_{\ra{ij}}+\mu)}+\frac{g_im_{ij}\alpha_{\ra{ij}}r_{ij}}{(r_{ij}-\mu)(\alpha_{\ra{ij}}+\mu)(\alpha_{\la{ij}}+\mu)}$. Based on equation~\ref{eq:res} and the fact that $N_i^r$ is fixed, we have:
    \begin{equation}
    N_{ii}^*+\sum_{j=1,j\neq i}^P K_{ij}N_{ii}^*=N_i^r
    \end{equation}
    We therefore have
    \begin{equation}
    N_{ii}^*=\frac{1}{1+\sum_{j=1,j\neq i}^P K_{ij}}N_i^r
    \end{equation}
	which gives equation~\ref{eq:equilibrium_ij}. 
	
	To prove that the equilibrium solution is G.A.S, one can write the whole matrix $M$ for system $dN/dt=MN$. Since the urban areas can be viewed as a connected network, it can be easily shown that diagonal entries in $M$ are strictly negative (either $-\mu$ or $-g_i+\mu$) and the matrix $M$ has all negative real eigenvalues, which implies that the equilibrium point is G.A.S. 
\end{proof}

In the following sections, we will write $N_{ij}$ to denote equilibrium population flow $N_{ij}^*$ for notation simplicity.

\section{Modeling disease spreading with travel contagion}
\subsection{Notation}
We summarize the list of variables used in this section as follows
\begin{table}[htbp!]
	\setstretch{1.2}
	\centering
	\caption{Table of notation}
	\label{notation}
	\begin{tabular}{p{2cm}p{13.5cm}}
		\hline
		Notation & Description \\ \hline
		&\textbf{Variables}\\ \hline
		$S_{ij}$ & Susceptible population who are residents of zone $i$ and currently in zone $j$. \\
		$E_{ij}$ & Exposed (latent) population who are residents of zone $i$ and currently in zone $j$.\\
		$I_{ij}$ & Infected population who are residents of zone $i$ and currently in zone $j$. \\
		$R_{ij}$ & Recovered population who are residents of zone $i$ and currently in zone $j$.\\
		$S_{\ra{ij}},S_{\la{ij}}$ & Susceptible population in travel from $i$ to $j$ and from $j$ to $i$ respectively. Similar notations are used for $E,I$ and $R$ population.\\
		$N_i^p$ & Population at zone $i$. \\
		$N_{ij}$ & The amount of people currently at zone $j$ who are the residents of zone $i$. \\
		\hline
		&\textbf{Fixed parameters}\\ \hline
		$D$ & Total number of travel modes available. \\
		$d$ & Travel mode $d$, where $d=1,2,...,D$. \\
		$\beta^A$ & Disease transmission rate per valid contact at the activity location. \\
		$\beta_{d}^T$ & Disease transmission rate per valid contact during travel using mode $d$. \\ 
		$c_{ij}^d$ & The ratio of people who choose travel mode $d$ between zone $i$ and zone $j$. \\
		$1/\sigma$ & The expected latent duration remaining in $E$ before moving to $I$. \\ 
		$1/\gamma$ & The expected recover duration remaining in $I$ before moving to $R$. \\
		$\kappa_{ij}$ & Expected number of valid contacts for residents of $i$ who are currently at zone $j$. \\ 
		$\kappa_{ij,kl}^d$ & Expected number of valid contacts for travelers from $i$ to $j$ who come across with travelers from $k$ to $l$ using the same travel mode $d$. \\ 
		
		\hline
	\end{tabular}
\end{table}

\subsection{System dynamics}
Important components missing from the previously discussed SEIR model are the spatial movement of urban population and the transmission of infectious diseases due to contacts during travel. The model only considers local dynamics, but it is essential for urban areas to model explicitly how population flow moving around the city. In particular, these flows are driven by various activities, such as work, school, or entertainment. And people get in contact with others by taking different activities through various transportation tools. As long as some individuals are infected, their activities and the urban transportation mobility will take the disease to every corner of the city. This motivates us to understand the the spread of infectious disease in urban area by modeling the system with the following 6 dynamics: 
\begin{enumerate}
	\item The mobility dynamics of urban population follows the model as discussed in Section~\ref{sec:mobility}.
	\item Similar to the SEIR model, we consider $S_{ij}$ being affected by $I_{ij}$ with the inner-zone activity contagion rate $\beta^{A}$. 
	\item In addition to the spread of diseases due to activity engagement, we also consider people get infected during travel. That is, $S_{ij}$ may be infected by contacting with $I_{ij}$ if they use the same travel mode $d$, with the contagion rate of $\beta^T_d$. 
	\item Once people in $S_{ij}$ are infected, they become $E_{ij}$. They are not infectious until the end of the latent period, and the length of latent period is characterized by $1/\sigma$.
	\item People in $E_{ij}$ become $I_{ij}$ at the end of the latent period. And the length of the infectious period is characterized by $1/\gamma$. 
	\item At the end of their infectious period, people in $I_{ij}$ become $R_{ij}$. For simplicity, we consider that they gain permanent immunity to the disease (e.g., death, vaccinated, or fully treated) and will no longer be infected.  
\end{enumerate}

In particular, the contagion process between susceptible and infected population can be illustrated by the example of a 3-zone network as shown in Figure~\ref{fig:examp}
\begin{figure}[h!]
	\centering
	\includegraphics[width=150mm]{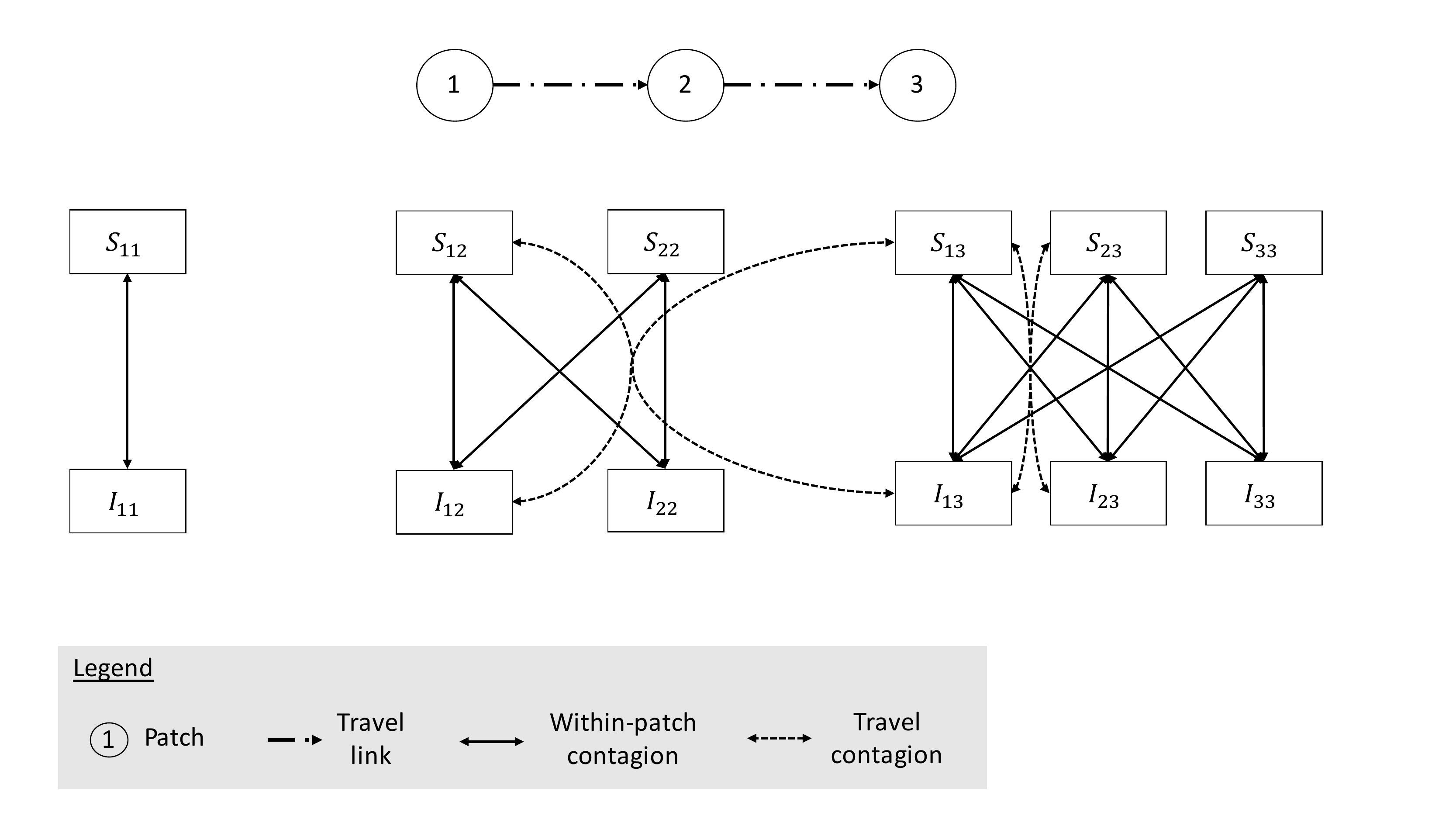}
	\caption{Illustration of the disease contagion process in a 3-zone network}
	\label{fig:examp}
\end{figure}
For the illustration network, there are two travel links 1-2 and 2-3 and three possible routes: from zone 1 to zone 2, from zone 1 to zone 3, and from zone 2 to zone 3. And we have two different contagions in this network. One is the inner-zone contagion which takes place between $S$ and $I$ population in the same zone (1-3). Additionally, there is travel contagion (shown in dashed line), which may happen if two population share overlay segments in their travel routes. Consequently, we have two possible travel contagions: between people travel from 1 to 2 and 1 to 3, between people travel from 1 to 3 and 2 to 3, and vice versa. While this is only an illustration, the overlay segments should also be distinguished depending on the particular mode chosen by the commuters. 

\subsection{Formulation}
With the above discussions on the system dynamics, we first differentiate the two contagion dynamics: activity contagion and travel contagion, and can be written mathematically as:
\begin{equation}
f_{ij}(S,I)=\beta^{A}_j\kappa_{ij}\frac{S_{ij}(I_j+\sum_{k=1}^{N}I_{kj})}{N_j^p}
\end{equation}
\begin{equation}
f_{i}(S,I)=\beta^{A}_i\kappa_{i}\frac{S_{i}(I_i+\sum_{k=1}^{N}I_{ki})}{N_i^p}
\end{equation}
\begin{equation}
h_{\la{ij}}(S,I)=\sum_{d=1}^D c_{ij}^{d} \beta_d^T {S}_{\la{ij}} [\sum_{k=1}^P\sum_{l=1}^P \frac{\kappa_{\la{kl},\la{ij}}^d I_{\la{kl}}}{N_{\la{kl}}}+\frac{\kappa_{\ra{kl},\la{ij}}^d I_{\ra{kl}}}{N_{\ra{kl}}}]
\label{eq:travel_return_infect}
\end{equation}
\begin{equation}
h_{\ra{ij}}(S,I)=\sum_{d=1}^D c_{ij}^{d} \beta_d^T {S}_{\ra{ij}} [\sum_{k=1}^P\sum_{l=1}^P \frac{\kappa_{\la{kl},\ra{ij}}^d I_{\la{kl}}}{N_{\la{kl}}}+\frac{\kappa_{\ra{kl},\ra{ij}}^d I_{\ra{kl}}}{N_{\ra{kl}}}]
\label{eq:travel_out_infect}
\end{equation}

Equation~\ref{eq:travel_return_infect} summarizes the S population who returns to home $i$ and get infected during travel, by encountering I population to-and-from resident location in other places. Specifically, $c_{ij}^dg_im_{ij}S_{ii}$ captures the ratio of $S_{ii}$ who leaves for destination $j$ using transportation mode $d$. $\frac{\kappa_{kl,ij}^d I_{kl}}{N_{kl}}$ denotes the contact rate with infectious population returning from $l$ to $k$, which is simplified from:
\begin{equation}
\frac{\kappa_{kl,ij}^d c_{kl}^dr_{kl} I_{kl}}{c_{kl}^dr_{kl} N_{kl}}
\end{equation}
and $\kappa_{kl,ij}^d$ refers to the number of contacts on average between ODs $kl$ and $ij$, and is a function of the travel time. Similarly, $\frac{\kappa_{kl,ij}^d I_{kk}}{N_{kk}}$ represents the contact rate with the infectious population leaving from $k$ to $l$, which is simplified from:
\begin{equation}
\frac{\kappa_{kl,ij}^d c_{kl}^dg_km_{kl}I_{kk}}{c_{kl}^dg_km_{kl}N_{kk}}
\end{equation}
With the above four different incidences for possible disease transmission, we can formally express the mathematical formulations that describe the transportation disease spreading process in urban area as follows:
\begin{eqnarray}
\begin{aligned}
\frac{E_i}{dt}=&-g_{i}E_{i}+\sum_{j=1}^{P} \alpha_{\la{ij}}E_{\la{ij}}-\sigma E_{i}+f_{i}(S,I)-\mu E_{i}\\
\frac{E_{ij}}{dt}=&-r_{ij}E_{ij}+\alpha_{\ra{ij}}E_{\ra{ij}}-\sigma E_{ij}+f_{ij}(S,I)-\mu E_{ij}\\
\frac{{E}_{\ra{ij}}}{dt}=&g_{i}m_{ij}E_{i}+h_{\ra{ij}}(S,I)-\alpha_{\ra{ij}}E_{\ra{ij}}-\sigma E_{\ra{ij}}-\mu E_{\ra{ij}}\\
\frac{{E}_{\la{ij}}}{dt}=&r_{ij}E_{ij}+h_{\la{ij}}(S,I)-\alpha_{\la{ij}}E_{\la{ij}}-\sigma E_{\la{ij}}-\mu E_{\la{ij}}\\
\frac{I_i}{dt}=&-g_{i}I_{i}+\sum_{j=1}^{P}\alpha_{\la{ij}} I_{\la{ij}}+\sigma E_{i}-\gamma I_{i}-\mu I_{i}\\
\frac{I_{ij}}{dt}=&-r_{ij}I_{ij}+\alpha_{\ra{ij}}I_{\ra{ij}}+\sigma E_{ij}-\gamma I_{ij}-\mu I_{ij}\\
\frac{{I}_{\ra{ij}}}{dt}=&g_{i}m_{ij}I_{i}-\alpha_{\ra{ij}}I_{\ra{ij}}+\sigma E_{\ra{ij}}-\gamma I_{\ra{ij}}-\mu I_{\ra{ij}}\\
\frac{{I}_{\la{ij}}}{dt}=&r_{ij}I_{ij}-\alpha_{\la{ij}}I_{\la{ij}}+\sigma E_{\la{ij}}-\gamma I_{\la{ij}}-\mu I_{\la{ij}}\\
\frac{S_i}{dt}=&-g_{i}S_{i}+\sum_{j=1}^{P}\alpha_{\la{ij}} S_{\la{ij}}-f_{i}(S,I)+\mu(N_i^r-S_{i})\\
\frac{S_{ij}}{dt}=&-r_{ij}S_{ij}+\alpha_{\ra{ij}}S_{\ra{ij}}-f_{ij}(S,I)-\mu S_{ij}\\
\frac{{S}_{\ra{ij}}}{dt}=&g_{i}m_{ij}S_{i}-h_{\ra{ij}}(S,I)-\alpha_{\ra{ij}}S_{\ra{ij}}-\mu S_{\ra{ij}}\\
\frac{{S}_{\la{ij}}}{dt}=&r_{ij}S_{ij}-h_{\la{ij}}(S,I)-\alpha_{\la{ij}}S_{\la{ij}}-\mu S_{\la{ij}}\\
\frac{R_i}{dt}=&-g_{i}R_{i}+\sum_{j=1}^{P}\alpha_{\la{ij}} R_{\la{ij}}+\gamma I_{i}-\mu R_{i}\\
\frac{R_{ij}}{dt}=&-r_{ij}R_{ij}+\alpha_{\ra{ij}}R_{\ra{ij}}+\gamma I_{ij}-\mu R_{ij}\\
\frac{{R}_{\ra{ij}}}{dt}=&g_{i}m_{ij}R_{i}-\alpha_{\ra{ij}}R_{\ra{ij}}+\gamma I_{\ra{ij}}-\mu R_{\ra{ij}}\\
\frac{{R}_{\la{ij}}}{dt}=&r_{ij}R_{ij}-\alpha_{\la{ij}}R_{\la{ij}}+\gamma I_{\la{ij}}-\mu R_{\la{ij}}\\
\end{aligned}
\label{eq:seir_system}
\end{eqnarray}

The system described by equation~\ref{eq:seir_system} can be easily understood by decomposing each equation into the part of mobility dynamics and the part of disease dynamics. For the mobility dynamics, the system of equations is consistent with the mobility model which captures the movement of $S,E,I,R$ population across the city as well as the death and birth (only into susceptible population). As for the disease dynamics, it involves the migration of population between compartments following the contagion process, the end of latent period and the end of infectious period.

\begin{prop}
	Based on equation~\ref{eq:seir_system}, we can decompose the system into $\mathcal{F}-\mathcal{V}$ where we have
	\begin{equation}
	\mathcal{F}=
	\begin{bmatrix}
	f_{i}\\
	f_{ij}\\
	h_{\ra{ij}}\\
	h_{\la{ij}}\\
	0\\
	0\\
	0\\
	0
	\end{bmatrix}
	, \mathcal{V}=
	\begin{bmatrix}
	(g_{i}+\mu+\sigma)E_{i}-\sum_{j=1}^{P} \alpha_{\la{ij}}E_{\la{ij}}\\
	(r_{ij}+\mu+\sigma)E_{ij}-\alpha_{\ra{ij}}E_{\ra{ij}} \\
	(\mu+\sigma+\alpha_{\ra{ij}})E_{\ra{ij}}-g_{i}m_{ij}E_{i}\\
	(\mu+\sigma+\alpha_{\la{ij}})E_{\la{ij}}-r_{ij}E_{ij}\\
	(g_{i}+\mu+\gamma)I_{i}-\sum_{j=1}^{P} I_{\la{ij}}-\sigma E_{i}\\
	(r_{ij}+\mu+\gamma)I_{ij}-\alpha_{\ra{ij}}I_{\ra{ij}}-\sigma E_{ij}\\
	(\mu+\gamma+\alpha_{\ra{ij}})I_{\ra{ij}}-g_{i}m_{ij}I_{i}-\sigma E_{\ra{ij}}\\
	(\mu+\gamma+\alpha_{\la{ij}})I_{\la{ij}}-r_{ij}I_{ij}-\sigma E_{\la{ij}}\\
	\end{bmatrix}
	\end{equation}
\end{prop}

Considering that an infectious disease has been introduced into a city and follows the system described in equation~\ref{eq:seir_system}, an essential question to be answered is that if the disease will eventually invade the population. Specifically, the above system compromises two equilibrium points as following: 
\begin{defn}
	The disease dynamic system characterized by equation 5.1-5.8 has two equilibrium points. The first equilibrium point $x_0$ is the disease free equilibrium (DFE):
	\begin{equation}
	S_{ij}=N_{ij},E_{ij}=I_{ij}=R_{ij}=0
	\end{equation}
	\begin{equation}
	x_0=(N,0,0,0)
	\end{equation}
	The second equilibrium point is the endemic equilibrium with:
	\begin{equation}
	\sum_{ij} I_{ij}>0
	\end{equation}
	such that a strictly positive fraction population will be in the infectious state.  
\end{defn}

To address the previous question, it requires the understanding of the stability of disease free equilibrium (DFE). In particular, if the DFE is stable, then the disease will be absent from the population, otherwise it is always possible for disease outbreak. We next discuss the stability of the DFE with the presented Trans-SEIR model.

\subsection{Model analysis}
To better analyze the model properties, we first rearrange the modeling parameters in the following order:
\[
E_{1},E_{11}, E_{\ra{11}},E_{\la{11}}, E_{12}, E_{\ra{12}},E_{\la{12}},\dots,E_{\la{nn}},I_{1},I_{11} I_{\ra{11}},I_{\la{11}},\dots
\]
with $E_{\la{ii}},I_{\la{ii}},E_{\ra{ii}},I_{\ra{ii}}$ being strictly zero as we do not consider within zone travels. We denote $Z_i^E$ and $Z_i^I$ as the set of $E$ and $I$ population compartments that are associated with residents of zone $i$, e.g., $Z_i^E=\{E_{11}, E_{\ra{11}},E_{\la{11}}, E_{12},...,E_{\la{nn}}\}$. 

If we linearize the system at the DFE point $x_0$, we have $D\mathcal{F}(x_0)$ as:
\begin{equation}
D\mathcal{F}(x_0)=\begin{bmatrix}
0 & 0 & \frac{df_{ii}(x_0)}{d I_{ii}}+\frac{dh_{ii}(x_0)}{d I_{ii}} & \frac{df_{ii}(x_0)}{d I_{ij}}+\frac{dh_{ii}(x_0)}{d I_{ij}} \\
0 & 0 &\frac{df_{ij}(x_0)}{d I_{ii}}+\frac{dh_{ij}(x_0)}{d I_{ii}} & \frac{df_{ij}(x_0)}{d I_{ij}}+\frac{dh_{ij}(x_0)}{d I_{ij}} \\
0 & 0 & 0 & 0\\
0 & 0 & 0 & 0 \\
\end{bmatrix}
\end{equation}

We can further partition matrix $V$ as in the following format:
\begin{equation}
V=\begin{bmatrix}
V_E & \bigzero \\
-\Sigma & V_I\\
\end{bmatrix}=\left[\begin{array}{@{}c|c@{}}
\begin{matrix}
V_{E}(1) &  &  \\
& \ddots & \\
& & V_{E}(n)
\end{matrix}
& \bigzero \\
\hline
-\Sigma &
\begin{matrix}
\begin{matrix}
V_{I}(1) &  &  \\
& \ddots & \\
& & V_{I}(n)
\end{matrix}
\end{matrix}
\end{array}\right]
\end{equation}

where $\Sigma$ represents the diagonal matrix of $\sigma$. We use $V_E(i)$, with dimension $3n+1\times3n+1$, to represent the block Jacobian matrix with respect to elements the corresponding population compartments in $\{i1,i2,...,in\}$. We can formally express each block of the corresponding matrix as:

\begin{equation}
V_E(i)=
\tiny
\begin{blockarray}{ccccccccccc}
\small
E_{i1} & E_{\ra{i1}} &  E_{\la{i1}} & \dots & E_{i} & \dots & E_{in} & E_{\ra{in}} &  E_{\la{in}} \\
& & & & & & & & & &\\
\begin{block}{[cccccccccc]l}
r_{i1}+\mu+\sigma &-\alpha_{\ra{i1}}& & & & & & & & & E_{i1}\\
&  \alpha_{\ra{i1}}+\mu+\sigma & &  & -g_im_{i1} & & & & & & E_{\ra{i1}}\\ 
-r_{i1} & &  \alpha_{\la{i1}}+\mu+\sigma& & \vdots & & & & & & E_{\la{i1}} \\
& & &  \ddots   & & & & & & & \vdots\\
& & -\alpha_{\la{i1}}& \dots & g_i+\mu+\sigma  & \dots & & &  -\alpha_{\la{in}} & & E_{i}\\
& & & &\vdots &  \ddots & & & & & \vdots\\
& & & & & &  r_{in}+\mu+\sigma &-\alpha_{\ra{in}} &  & & E_{in}\\
& & & &-g_im_{in} & & &   \alpha_{\ra{in}}+\mu+\sigma & & &  E_{\ra{in}} \\
& & & & & &-r_{in} &  &   \alpha_{\la{in}}+\mu+\sigma & & E_{\la{in}}\\
\end{block}
\end{blockarray}
\end{equation}

Finally, we have the same formulation for the corresponding block diagonal matrices ($V_I(i)$) of $I$, with the only differences being that all $\sigma$ are replaced by $\gamma$. As a consequence, $V$ is a lower-triangular matrix. 

As for the transmission part of the formulation which corresponds to the new generation of infectious population, we should have:

\begin{equation}
\mathcal{F}=\left[\begin{array}{@{}c|c@{}}
\bigzero &  F+H\\
\hline
\bigzero & \bigzero 
\end{array}\right]
\end{equation}
In the equation, $F$ represents the disease transmission that are related to local compartments and $H$ accounts for all transmissions that happen during travel. In this regard, we can also break $F+H$ into $n^2$ blocks of matrix with each block $(F+H)(Z_i^E,Z_j^I)$ being a $3n\times 3n$ matrix of the following form: 

\begin{equation}
F+H=\begin{bmatrix}
F+H(Z_1^E,Z_1^I) & \dots & F+H(Z_1^E,Z_n^I) \\
F+H(Z_2^E,Z_1^I) & \dots & F+H(Z_2^E,Z_n^I)  \\
\vdots & \ddots & \vdots\\
F+H(Z_n^E,Z_1^I) & \dots & F+H(Z_n^E,Z_n^I) 
\end{bmatrix},
\end{equation}

\begin{equation}
(F+H)(Z_i^E,Z_j^I)=
\tiny
\begin{blockarray}{cccccccc}
\small
I_{j1} & I_{\ra{j1}} &  I_{\la{j1}} & \dots & I_{jn} & I_{\ra{jn}} &  I_{\la{jn}} \\
& & & & & & \\
\begin{block}{[ccccccc]l}
\beta_1^A\kappa_{i1}\frac{N_{i1}}{N_1^p} & & &  & & & &   \dot{E}_{i1}\\
& h_{\ra{i1},\ra{j1}}& h_{\ra{i1},\la{j1}} & \dots & & h_{\ra{i1},\ra{jn}}& h_{\ra{i1},\la{jn}} & \dot{E}_{\ra{i1}}\\ 
& h_{\la{i1},\ra{j1}} & h_{\la{i1},\la{j1}} & \dots & & h_{\la{i1},\ra{jn}} & h_{\la{i1},\la{jn}} & \dot{E}_{\la{i1}} \\
& & &  \ddots   & & & & \vdots\\
& & &  & \beta_n^A\kappa_{in}\frac{N_{in}}{N_n^p} & & &  \dot{E}_{in}\\
& h_{\ra{in},\ra{j1}}& h_{\ra{in},\la{j1}}& \dots & & h_{\ra{in},\ra{jn}}& h_{\ra{in},\la{jn}} & \dot{E}_{\ra{in}}\\
& h_{\la{in},\ra{j1}} & h_{\la{in},\la{j1}} & \dots & & h_{\la{in},\ra{jn}} & h_{\la{in},\la{jn}} & \dot{E}_{\la{in}}\\
\end{block}
\end{blockarray}
\end{equation}

and each $h$ being represented as:
\begin{equation}
\begin{aligned}
h_{\la{ij},\ra{kl}}=\sum_{d=1}^D c_{ij}^{d} \beta_d^T\kappa_{\ra{kl},\la{ij}}^d  \frac{{N}_{\la{ij}}}{N_{\ra{kl}}}, \quad h_{\ra{ij},\ra{kl}}=\sum_{d=1}^D c_{ij}^{d} \beta_d^T\kappa_{\ra{kl},\ra{ij}}^d  \frac{{N}_{\ra{ij}}}{N_{\ra{kl}}}\\
h_{\la{ij},\la{kl}}=\sum_{d=1}^D c_{ij}^{d} \beta_d^T\kappa_{\la{kl},\la{ij}}^d  \frac{{N}_{\la{ij}}}{N_{\la{kl}}},\quad h_{\ra{ij},\la{kl}}=\sum_{d=1}^D c_{ij}^{d} \beta_d^T\kappa_{\la{kl},\ra{ij}}^d  \frac{{N}_{\ra{ij}}}{N_{\la{kl}}}\\
\end{aligned}
\end{equation}

\begin{prop}
	$V_E(i)$ and $V_I(i)$ in $V$ are invertible and $V_E(i)^{-1}$ and $V_I(i)^{-1}$ are strictly positive. 
	\label{prop:inverse}
\end{prop}
\begin{proof}
	Since $V_E(i)$ and $V_I(i)$ have similar structure, without loss of generality, we here only show that $V_E(i)$ is invertible and same argument can be applied to $V_I(i)$.
	
	We note that, for $V_E(i)$, all diagonal elements are strictly positive and all off diagonal elements are either zero or strictly negative. In this regard, $V_E(i)$ satisfies the Z-sign pattern. In addition, we have the column sum for each column of $V_E(i)$ being $\mu+\sigma$, which is also strictly positive. This suggests that all eigenvalues of $V_E(i)$ have positive real part. This combines with the Z-sign pattern imply that $V_E(i)$ is a non-singular M-matrix and this leads to $V_E(i)^{-1}>0$. 
\end{proof}

Based on Proposition~\ref{prop:inverse}, we can write the inverse of $V$ as following:
\begin{equation}
V^{-1}=\begin{bmatrix}
V_E^{-1} & 0 \\
V_I^{-1}\Sigma V_E^{-1} & V_I^{-1}\\
\end{bmatrix}
\end{equation}

And $FV^{-1}$ is expressed as:
\begin{equation}
FV^{-1}=\begin{bmatrix}
(F+H) V_I^{-1}\Sigma V_E^{-1} & HV_I^{-1} \\
0 & 0
\end{bmatrix}
\label{eq:r0final}
\end{equation}

As the (1,1) entry in equation~\ref{eq:r0final} denotes the number of secondary infections that a new entrance in $E$ will produce, we can therefore determine the basic reproduction number ($R_0$) by
\begin{equation}
R_0=\sigma \rho((F+H) (V_EV_I)^{-1})
\label{eq:R0_value}
\end{equation}
where $\rho$ refers to the spectral radius (the real part of the largest eigenvalue) and $(V_EV_I)^{-1}$ is a block diagonal matrix. 
\begin{prop}
	The average number of secondary infectious due to a single infectious person into a fully susceptible population is given by $R_0=\rho\{(F+H)V^{-1}\}$, with $\rho$ denoting the spectral radius. In addition, if $R_0<1$ then the DFE is locally asymptotically stable. Otherwise the DFE is unstable and the disease will invade the population. 
	\label{prop:r0}
\end{prop}

We omit the proof of Proposition~\ref{prop:r0} as it follows directly from Theorem 2 in~\cite{van2002reproduction}. Note that the each entry in $F+H$ has specific physical interpretations, with the $(ij,kl)$th entry representing the new infections that an additional infectious person in $kl$ compartment will produce in the $ij$ compartment. Consequently, each eigenvalue corresponds to the column of $(F+H)V^{-1}$ refers to the $R_0$ value of that particular compartment, and that the $R_0^i$ for zone $i$ is determined by the largest $R_0$ value among all compartments that are associated with the residents of $i$.

We further observe that the column sum of $V_E$ is $\mu+\sigma$, and the column sum of $V_I$ is $\mu+\gamma$. As a result, the row sum of each column in $(V_EV_I)^{-1}$ is given by $\frac{1}{(\mu+\gamma)(\mu+\sigma)}$. Since $(V_EV_I)$ is a block diagonal matrix, so is $(V_EV_I)^{-1}$. Then for the $ij$th column of the $i$th block, only the entries in rows $i1,...,\la{in}$ may be non-zero and all other entries of the column are strictly 0. Then we can write out the vector for column sum $C_{ij}$ explicitly as:
\begin{eqnarray}
\begin{aligned}
C_{ij}&=\mathbbm{1}^{\intercal}(\sum_i f_i (V_EV_I)^{-1}_{ij}+h (V_EV_I)^{-1}_{ij})\\
& = \beta_1^A\kappa_{11}\frac{N_{11}}{N_1^p}v_{ij,i1}+\dots+\beta_n^A\kappa_{1n}\frac{N_{1n}}{N_n^p}v_{ij,in}+\beta_1^A\kappa_{21}\frac{N_{21}}{N_1^p}v_{ij,i1}+\dots+\beta_n^A\kappa_{2n}\frac{N_{2n}}{N_n^p}v_{ij,in}+\dots+ \beta_n^A\kappa_{nn}\frac{N_{nn}}{N_n^p}v_{ij,in}\label{zone_inf}\\
&\,+\sum_{m=1}^n(\sum_{k=1}^n\sum_{l=1}^n h_{\ra{kl},\ra{im}}+h_{\la{kl},\ra{im}})v_{ij,\ra{im}}+(\sum_{k=1}^n\sum_{l=1}^n h_{\ra{kl},\la{im}}+h_{\la{kl},\la{im}})v_{ij,\la{im}}\\
& =\sum_{m=1}^n (\beta_m^A\kappa_{1m}\frac{N_{1m}}{N_m^p}+\dots+\beta_m^A\kappa_{nm}\frac{N_{nm}}{N_m^p})v_{ij,im}\\
&\,+\sum_{m=1}^n (\sum_{k=1}^n\sum_{l=1}^n h_{\ra{kl},\ra{im}}
+h_{\la{kl},\ra{im}})v_{ij,\ra{im}}+(\sum_{k=1}^n\sum_{l=1}^n h_{\ra{kl},\la{im}}+h_{\la{kl},\la{im}})v_{ij,\la{im}}\\
\end{aligned}
\end{eqnarray}

As the column sum of $C_{ij}$ equals $\frac{1}{(\mu+\gamma)(\mu+\sigma)}$, we then have
\begin{equation}
\sum_{m=1}^n v_{ij,im} = \frac{1}{(\mu+\gamma)(\mu+\sigma)}
\end{equation}

Denote $\beta_{min}^A\kappa_{min}$ and $\beta_{max}^A\kappa_{max}$ be the the greatest of all activity compartments, e.g., 
\begin{equation}
\beta_{max}^A\kappa_{max}=\max_{i,j=1,...,n} \beta_i^A\kappa_{ji},\quad \beta_{min}^A\kappa_{min}=\min_{i,j=1,...,n} \beta_i^A\kappa_{ji}
\end{equation}
and $h_{max}$ and $h_{min}$ as the the trip segment with the highest and lowest contagious rate:
\begin{equation}
h_{max}=\max\{\max_{i,j=1,...,n} \sum_{k=1}^n\sum_{l=1} h_{\ra{kl},\ra{ij}}
+h_{\la{kl},\ra{ij}},\max_{i,j=1,...,n} \sum_{k=1}^n\sum_{l=1} h_{\ra{kl},\la{im}}+h_{\la{kl},\la{im}}\}
\end{equation}
\begin{equation}
h_{min}=\min\{\min_{i,j=1,...,n} \sum_{k=1}^n\sum_{l=1} h_{\ra{kl},\ra{ij}}
+h_{\la{kl},\ra{ij}},\min_{i,j=1,...,n} \sum_{k=1}^n\sum_{l=1} h_{\ra{kl},\la{im}}+h_{\la{kl},\la{im}}\}
\end{equation}
Then we have \begin{equation}
\frac{\min \{\beta_{min}^A\kappa_{min},h_{min}\}}{(\mu+\sigma)(\mu+\gamma)}  \leq C \leq  \frac{\max \{\beta_{max}^A\kappa_{max},h_{max}\}}{(\mu+\sigma)(\mu+\gamma)}
\end{equation}
and
\begin{equation}
\frac{\sigma \min \{\beta_{min}^A\kappa_{min},h_{min}\}}{(\mu+\sigma)(\mu+\gamma)}  \leq R_0\leq   \frac{\sigma \max \{\beta_{max}^A\kappa_{max},h_{max}\}}{(\mu+\sigma)(\mu+\gamma)}
\end{equation}

\begin{cor}
	Let $\beta_{max}^A\kappa_{max}$ be the highest within zone contagion rate of all zones, and $\beta_{min}^A\kappa_{min}$  be the lowest within zone contagion rate of all zones. If we assume $\kappa_{ij}$ being the same at all activity locations, then $R_0$ is bounded by the travel segments with the highest contagion rate. 
	\label{col:R0}
\end{cor} 


\section{Control the spread of disease within urban transportation system}
With the Trans-SEIR model capturing the mobility and disease dynamics, we next investigate the research question on the control of infectious diseases by placing entrance screening for transportation system. In particular, we have seen the implementation of radiation thermometers during the 2003 SARS, 2012 MERS and 2020 COVID-19 outbreaks at subway and bus stations and major transportation hubs. The entrance screening aims to identify passengers at risk with abnormal body temperature or susceptible symptoms. And additional measures such as distributing masks and random inspection by medical workers at the entrance of transportation systems may also be efficient to reduce the risk exposure during travel. While entrance screening can be of great value for the preventative purpose before the actual outbreak or to slow down the invasion of diseases during an outbreak, there is a trade-off among the success rate of screening, the amount of travelers that need to be examined and the incurred externalities such as excessive entry delays and limited manpower and medical resources. In this regard, the essential issue to be addressed is how to optimally distribute the available resources over the urban transportation systems to effectively curb the spread of infectious diseases. 

For the real-world transportation system, it is hardly possible to perform the entry control for all travel modes. Instead, it is only viable to conduct such control over the medium and high capacity travel modes, such as buses and metros. As a consequence, in our study, we assume that the travel modes are divided into three categories:
\begin{enumerate}
	\item \textbf{Low capacity mode} such as private vehicles and taxis
	\item \textbf{Medium capacity mode} such as vans and buses
	\item \textbf{High capacity mode} such as metro system
\end{enumerate}
We do not consider the low capacity mode as the control target, and the low capacity mode typically has fewer passengers per vehicle and therefore constituting the lowest chance of getting infected. On the contrary, passengers using medium or high capacity mode are usually exposed to more co-riders in an enclosed compartment for longer trip duration, and are therefore prone to higher chance of being infected. To frame the optimal entrance control problem, we introduce additional parameters $\epsilon$ and $\xi$, where $\xi$ denotes the success rate of screening of infectious travelers and $\epsilon_{ij}$ represents the fraction of passengers traveling between $i$ and $j$ that will be screened, with $\epsilon_{ij}\in[0,1]$. For simply, here we only consider that the same $\xi$ and $\epsilon_{ij}$ will be applied to all different modes, which can be easily extended to account for mode-specific control parameters. 

As the first step of framing the optimal control problem, we focus on scenario without resource constraint and entrance control alters the disease dynamics by reducing the number of infectious passengers entering the travel compartments $I_{\ra{ij}}$ and $I_{\la{ij}}$ once they are identified to be infected at the entrance. We consider that the identified infectious passengers will be permanently quarantined until they are no longer infectious, so that they join the recovered travel compartments $R_{\ra{ij}}$ and $R_{\la{ij}}$ respectively and shall no longer produce further infections. In this regard, we can rewrite part of the disease system as: 
\begin{eqnarray}
\begin{aligned}
\frac{{I}_{\ra{ij}}}{dt}=&(1-\epsilon_{\ra{ij}})g_{i}m_{ij}I_{ii}-\alpha_{\ra{ij}}I_{\ra{ij}}+\sigma E_{\ra{ij}}-\gamma I_{\ra{ij}}-\mu I_{\ra{ij}}\\
\frac{{I}_{\la{ij}}}{dt}=&(1-\epsilon_{\la{ij}})r_{ij}I_{ij}-\alpha_{\la{ij}}I_{\la{ij}}+\sigma E_{\la{ij}}-\gamma I_{\la{ij}}-\mu I_{\la{ij}}\\
\frac{{R}_{\ra{ij}}}{dt}=&g_{i}m_{ij}(R_{ii}+\epsilon I_{ii})-\alpha_{\ra{ij}}R_{\ra{ij}}+\gamma I_{\ra{ij}}-\mu R_{\ra{ij}}\\
\frac{{R}_{\la{ij}}}{dt}=&r_{ij}(R_{ij}+\epsilon I_{ij})-\alpha_{\la{ij}}R_{\la{ij}}+\gamma I_{\la{ij}}-\mu R_{\la{ij}}\\
\end{aligned}
\label{eq:control}
\end{eqnarray}
so that the whole disease system is still well defined and the summation of all population still equals $N$. Without loss of generality, we consider the corresponding control effectiveness of the mode $m$ being $\epsilon^m$, then $\epsilon$ simply refers to
\begin{equation}
\epsilon_{ij}=\sum_{m}c_{ij}^m\epsilon^m
\end{equation}By controlling the system without resource limitation, we may let $\epsilon=1$, e.g. all travelers will be examined, and the resulting system dynamics can be decomposed into  $F-\tilde{V}$ where $\tilde{V}$ takes the form:

\begin{equation}
\tilde{V}=\begin{bmatrix}
V_E & \bigzero \\
-\Sigma & \tilde{V}_I\\
\end{bmatrix}=\left[\begin{array}{@{}c|c@{}}
\begin{matrix}
V_{E}(1) &  &  \\
& \ddots & \\
& & V_{E}(n)
\end{matrix}
& \bigzero \\
\hline
-\Sigma &
\begin{matrix}
\begin{matrix}
\tilde{V}_{I}(1) &  &  \\
& \ddots & \\
& & \tilde{V}_{I}(n)
\end{matrix}
\end{matrix}
\end{array}\right]
\end{equation}
Specifically, $\tilde{V}_{I}(i)$ has the same diagonal elements of $V_{I}(i)$ but all its off-diagonal elements are zero. In this regard, each $\tilde{V}_I(i)$ is also a diagonal matrix and hence of larger dominant eigenvalue as compared to $V_I(i)$, suggesting an overall shorter infection period for the travelers. Following the same approach as discussed in the previous section, we can derive the reproduction number after transit control as:
\begin{equation}
\tilde{R}_0=\rho(F\tilde{V}^{-1})
\end{equation}
With curbing the spread of infectious through transportation control as the primary goal, we seek to identify proper $\epsilon$ to minimize the $\tilde{R}_0$, and the gap between $R_0$ and $\tilde{R}_0$ is therefore the reduction by controlling travel contagion. We note, however, there is no clean and close representation to formulate the relationship between $\epsilon$ and $\tilde{R}_0$, which creates a barrier for devising optimal allocation strategy of limited resources. To overcome this barrier, we seek an alternate approach to avoid this complication and build the connection between the two variables through a modified next generation matrix where we have
\begin{equation}
F(\epsilon)=\left[\begin{array}{@{}c|c@{}}
\bigzero &  F+H\\
\hline
\bigzero & \bar{V_I}(\epsilon),
\end{array}\right],
\tilde{V}=\begin{bmatrix}
V_E & \bigzero \\
-\Sigma & \tilde{V}_I\\
\end{bmatrix}
\end{equation}

where $\tilde{V}_I$ contains only the diagonal elements of $V_I$ and $-\bar{V}_I$ includes all the off-diagonal entries of $V_I$:
\begin{equation}
\bar{V}_I(i,\epsilon)=\tiny
\begin{blockarray}{ccccccccccc}
\small
I_{i1} & I_{\ra{i1}} &  I_{\la{i1}} & \dots & I_{i} & \dots & I_{in} & I_{\ra{in}} &  I_{\la{in}} \\
& & & & & & & & & &\\
\begin{block}{[cccccccccc]l}
& \alpha_{\ra{i1}}& & & & & & & & & I_{i1}\\
&   & &  & (1-\epsilon_{\ra{i1}})g_im_{i1} & & & & & & I_{\ra{i1}}\\ 
(1-\epsilon_{\la{i1}})r_{i1} & &  & & \vdots & & & & & & I_{\la{i1}} \\
& & &  \ddots   & & & & & & & \vdots\\
& & \alpha_{\la{i1}}&   & \dots & & &  & \alpha_{\la{in}} & & I_{i}\\
& & & &\vdots &  \ddots & & & & & \vdots\\
& & & & & &   &\alpha_{\ra{in}} &  & & I_{in}\\
& & & &(1-\epsilon_{\ra{in}})g_im_{in} & & &    & & &  I_{\ra{in}} \\
& & & & & &(1-\epsilon_{\la{in}})r_{in} &  &    & & I_{\la{in}}\\
\end{block}
\end{blockarray}
\end{equation}

We denote $\tilde{K}$ as the next generation matrix following $F(\epsilon)\tilde{V}^{-1}$ and $K$ being the original next generation matrix. We have the following propositions
\begin{prop}
	$R_0>1$ if and only if $\tilde{R}_0>1$, $R_0$=0 if and only if $\tilde{R}_0=0$, and $R_0<1$ if and only if $\tilde{R}_0<1$
	\label{prop:equivalence}
\end{prop}
\begin{proof}
	We have that $R_0=1$ while $s(F-V)=0$, and $R_0>1$ if $s(F-V)>0$, and $R_0<1$ if $s(F-V)<0$. By deriving $\tilde{R}_0$, we see that the corresponding $F-\tilde{V}$ still holds. This completes the proof. 
\end{proof}

\begin{prop}
	A disease is controllable through entrance screening in transportation system if and only if $\rho(F(1)\tilde{V}^{-1})<1$. 
	\label{prop:control_disease}
\end{prop}
The proof for Proposition~\ref{prop:control_disease} follows from the results of Proposition~\ref{prop:equivalence}, where $\rho(F(1)\tilde{V}^{-1})$ suggests that $\tilde{R}_0<1$ when $\epsilon=1$ and $\tilde{R}_0<1$ is equivalent to $R_0<1$. Hence the proposition states that the disease is controllable through public transportation if we may have access to unlimited resources to completely quarantine all infectious travelers and reduce the $R_0$ lower than 1. 

In addition, Proposition~\ref{prop:equivalence} states that the we can focus on controlling $\tilde{R}_0$ instead of $R_0$ to mitigate the impact from infectious disease, and more importantly eliminating the disease if we can lower $\tilde{R}_0$ to be smaller than 1. Cautions should be exercised, however, as $\tilde{R}_0$ does not have the same physical meaning as $R_0$ and can not be interpreted as the number secondary infections. Based on this equivalency, we can rephrase Corollary~\ref{col:R0} to derive the corresponding value for $\tilde{R}_0$ by measuring the column sum for $F(\epsilon)\tilde{V}^{-1}$ as:
\begin{equation}
\begin{aligned}
&\tilde{C}_{ij}(\epsilon)=\mathbbm{1}^{\intercal}(\sum_i f_i (V_E\tilde{V}_I)^{-1}_{ij}+h (V_E\tilde{V}_I)^{-1}_{ij}+\bar{V}(V_E\tilde{V}_I)^{-1}_{ij})\\
& = \frac{\beta_1^A\kappa_{11}+\sum_{j=1}^p (1-\epsilon_{\ra{1j}})g_1m_{1j}}{g_1+\gamma+\mu}\frac{N_{11}}{N_1^p}\tilde{v}_{ij,i1}+\dots+\frac{\beta_n^A\kappa_{1n}+(1-\epsilon_{\la{1n}})r_{1n}}{r_{1n}+\gamma+\mu}\frac{N_{1n}}{N_n^p}\tilde{v}_{ij,in} \\
&+\frac{\beta_1^A\kappa_{21}+(1-\epsilon_{\la{21}})r_{21}}{r_{21}+\gamma+\mu}\frac{N_{21}}{N_1^p}\tilde{v}_{ij,i1}+\dots+\frac{\beta_n^A\kappa_{2n}+(1-\epsilon_{\la{2n}})r_{2n}}{r_{2n}+\gamma+\mu}\frac{N_{2n}}{N_n^p}\tilde{v}_{ij,in}+\dots+ \frac{\beta_n^A\kappa_{nn}+\sum_{j=1}^p (1-\epsilon_{\ra{nj}})g_nm_{nj}}{g_n+\gamma+\mu}\frac{N_{nn}}{N_n^p}\tilde{v}_{ij,in}\label{zone_inf}\\
&\,+\sum_{m=1}^n(\sum_{k=1}^n\sum_{l=1}^n \frac{h_{\ra{kl},\ra{im}}+h_{\la{kl},\ra{im}}+2\alpha_{\ra{im}}}{\alpha_{\ra{im}}+\gamma+\mu})\tilde{v}_{ij,\ra{im}}+(\sum_{k=1}^n\sum_{l=1}^n \frac{h_{\ra{kl},\la{im}}+h_{\la{kl},\la{im}}+2\alpha_{\la{im}}}{\alpha_{\la{im}}+\gamma+\mu})\tilde{v}_{ij,\la{im}}\\
& =\sum_{m=1}^n (\frac{\beta_m^A\kappa_{1m}+(1-\epsilon_{\la{im}})r_{im}}{r_{im}+\gamma+\mu}\frac{N_{1m}}{N_m^p}+\dots+\frac{\beta_m^A\kappa_{nm}+(1-\epsilon_{\la{nm}})r_{nm}}{r_{nm}+\gamma+\mu}\frac{N_{nm}}{N_m^p})\tilde{v}_{ij,im}\\
&\,+\sum_{m=1}^n(\sum_{k=1}^n\sum_{l=1}^n \frac{h_{\ra{kl},\ra{im}}+h_{\la{kl},\ra{im}}+2\alpha_{\ra{im}}}{\alpha_{\ra{im}}+\gamma+\mu})\tilde{v}_{ij,\ra{im}}+(\sum_{k=1}^n\sum_{l=1}^n \frac{h_{\ra{kl},\la{im}}+h_{\la{kl},\la{im}}+2\alpha_{\la{im}}}{\alpha_{\la{im}}+\gamma+\mu})\tilde{v}_{ij,\la{im}}\\
\end{aligned}
\end{equation}

where we use $\tilde{v}(ij,kl)$ to denote the corresponding the $(ij,kl)$ entry of $V_E^{-1}$. Specifically, all $\tilde{v}(ij,kl)$ values are known once the parameters for $V_E$ are set. Considering the case where the screening cost is linearly increasing with higher effort level $\epsilon{ij}$ as well as a greater number of travelers $I_{ij}$, in order to reduce the risk of infectious disease, our objective is therefore to optimally allocate a given resource $B$ such that the maximum eigenvalue of $F(\epsilon)\tilde{V}^{-1}$ is minimized. This can be well approximated by minimizing the maximum column sum $\tilde{C}_{ij}$ through optimally resource allocation. And we can therefore express the optimal resource allocation problem for controlling the disease through urban transportation system as following:
\begin{equation}
\min_{\epsilon} \max_{ij} \tilde{C}_{ij}(\epsilon)
\end{equation}
subject to
\begin{equation}
\begin{aligned}
&\epsilon_{ij}\geq 0, \forall{i,j}\\
& c\sum_{i,j} \epsilon_{\la{ij}}r_{ij}N_{ij}+ c\sum_{i}\sum_j \epsilon_{\ra{ij}}g_im_{ij}N_{i}\leq B,\forall{i,j}
\label{eq:resource_con}
\end{aligned}
\end{equation}
Specifically, equation~\ref{eq:resource_con} denotes the resource constraint as the total screened passengers across all travel segments should not exceed the available manpower.  
As a min-max linear programming problem, the optimal resource allocation problem can be easily transferred into an equivalent linear programming problem:
\begin{equation}
\min z
\end{equation}
subject to
\begin{equation}
\begin{aligned}
& z\geq \tilde{C}_{ij}(\epsilon), \forall{i,j}\\
&\epsilon_{ij}\geq 0, \forall{i,j}\\
& c\sum_{ij,i\neq j} \epsilon_{\la{ij}}r_{ij}N_{ij}+ c\sum_{i}\sum_j \epsilon_{\ra{ij}}g_im_{ij}N_{i}\leq B,\forall{i,j}
& 
\end{aligned}
\end{equation}
The system of equation can be easily solved using commercial solvers, and for the numerical experiments we use CPLEX in Matlab to obtain the optimally allocated resources.

\section{Results}
\subsection{Experiment setting}
The numerical experiments are conducted on a desktop with @3.5 GHz CPU and $32GB$ RAM. The codes are written in Matlab and the disease trajectories are simulated using the ODE 45 solver. The numerical experiments in this study consist of two parts. Based on toy networks as shown in Figure~\ref{fig:small_network}, the first part is to compare the resulting disease dynamics with and without modeling the urban transportation system, including the differences in disease trajectories, the impact on $R_0$ values and the city-wide synchronization of the diseases. We also investigate how disease outbreaks may evolve with varying urban structures and transportation systems. The networks have 5 nodes and 25 OD pairs, meaning that each node is accessible from the other. As each OD pair may have four states (S,E,I,R), the ODE system of this test network therefore has 300 state variables and 300 equations. For the second experiment, we focus on validating the effectiveness of the Trans-SEIR model with real-world networks. We choose New York City (NYC) as the study area and verify the model outputs when compared to the reported data during the COVID-19 outbreak. The NYC network is shown in Figure~\ref{fig:nyc_net}. NYC is divided into 15 zones and the network is established following the Google transit information. The 15 zones lead to a system of 2760 state variables and 2760 equations for the ODE system. 
\begin{figure}[ht]
    \centering
    \includegraphics[width=0.8\linewidth]{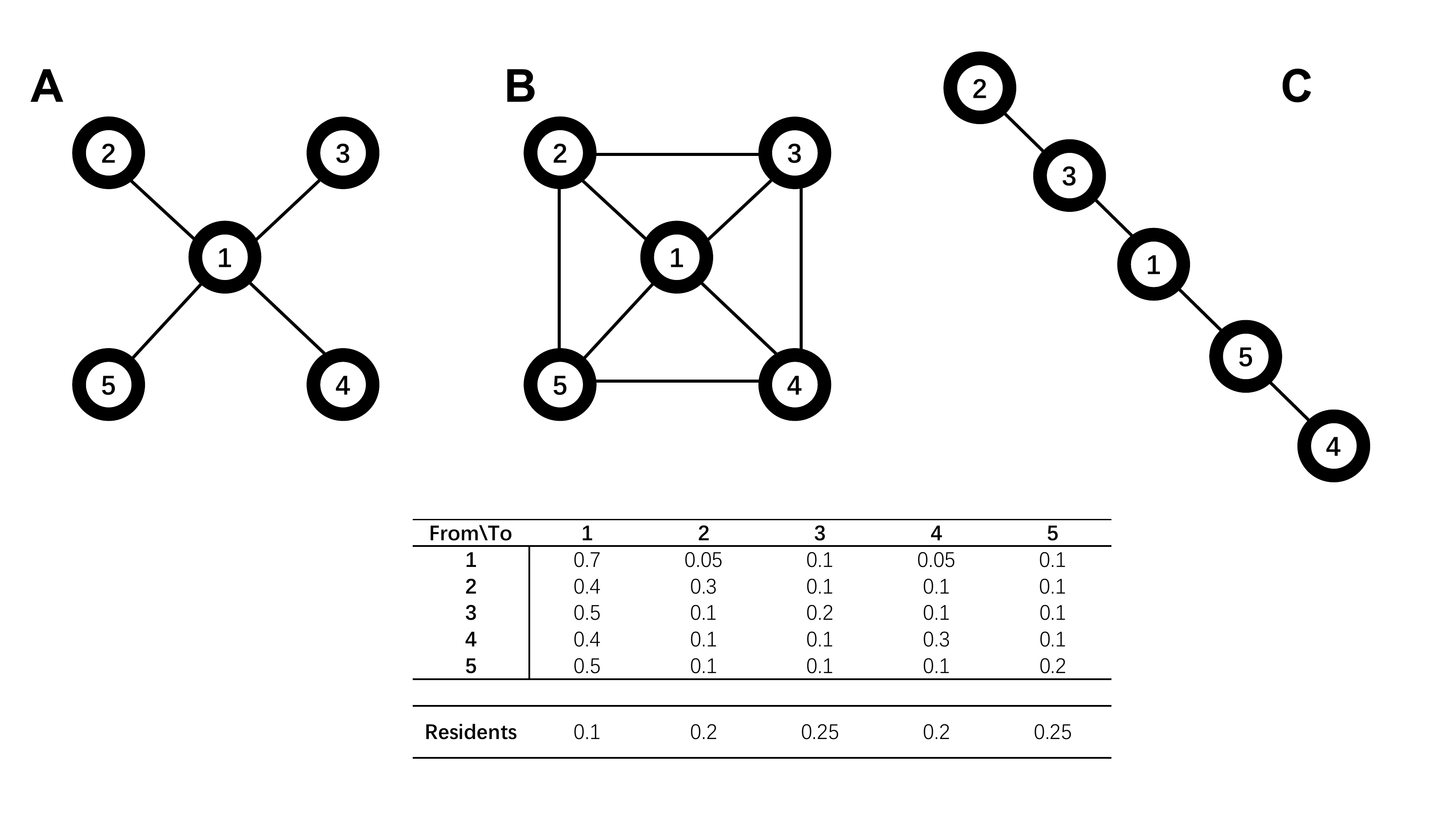}
    \caption{Network layouts of three sample networks}
    \label{fig:small_network}
\end{figure}

\subsection{Results on toy networks}
We first demonstrate the results of the Trans-SEIR model on the toy networks, and we focus on the validity of the disease threshold, the dynamics of the disease trajectories and the impacts of the travel related contagions. Figure~\ref{fig:sub_A} shows the disease of dynamics of two different scenarios in network A, where case 1 serves as the base scenario and case 2 doubles the transmission rates of those in case 1. Both cases have a total population of 10 million and start with an initial infectious population of 1,000 which are uniformly distributed among the residents of all five zones. Since case 1 has the $R_0$ value of 0.84, we observe that the number of infectious people decays exponentially fast and the system quickly moves towards the DFE state. On the contrary, case 2 has the $R_0$ value of 1.68 and exponential growth in the number of exposed and infectious populations can be immediately observed. The disease in case 2 eventually invades around 66.7\% of the total population before it dies out. While the complete trajectory was not shown here, we report that there is a strictly positive proportion of the population that remains infectious which agrees with the endemic equilibrium with $R_0>1$. For these two cases, we also compare the results with the outcomes when travel contagion is not considered (e.g., $\beta^T=0$). We can see that the two systems differ significantly and the results suggest that we may underestimate the risks of an urban disease outbreak if travel contagion is ignored, or equivalently we may overestimate the transmissibility of an infectious disease if conventional disease models are fitted to the historical data in urban areas. And the results therefore highlight the importance of considering travel contagion in transit systems for the most accurate understanding of urban disease dynamics and devising proper intervention and control strategies. Moreover, we find that, given the same population distribution and mobility patterns, the network structure of urban transit systems may have great impacts on the disease trajectories as shown in Figure~\ref{fig:three_city}. The network layouts of the three cities reflect different planning philosophies regarding the trade-off between accessibility and system efficiency. Network B represents the highest accessibility between trip ODs but also the highest operation costs with numerous direct transit routes. On the other hand, network C achieves the maximum transit usage with the minimum number of transit lines, but also leads to the highest number of passengers sharing overlapping trip hence the greatest risk of travel contagion. Based on the results of the toy network, we see that the number of secondary infections in network C may be 7.4\% higher than that of network B and 3\% higher than that of network A. As all these networks will reach the endemic equilibrium, the excessive contacts during travel in network C may increase the total infected population of network B by over 7.26\% and result in the arrival of the peak nearly 30 days earlier. And the early peak is primarily driven by the spatial synchronizations of the disease dynamics which is facilitated by the travel contagion, as shown in Figure~\ref{fig:sync}. In particular, we measure the rate of synchronization as the standard deviation of the change rates of the total infectious population of each zone. Higher std stands for the larger discrepancy among the zonal change rates and hence the disease dynamics are less synchronized. As the spread of disease usually starts with the first few external visitors in local areas, the rate of spatial synchronization therefore determines the pace and the scale of the citywide outbreak. The results in Figure~\ref{fig:sync} suggest that the disease dynamics across different zones are synchronized at an exponential rate, and the exponential decay rate for network C is 20\% higher than that without travel contagion or 10.2\% higher than that for network B, highlighting the catalyst effect of the urban transit system and travel contagion in promoting a faster pace for spatial synchronization. While urban transit systems have long been designed as an affordable mobility solution for the mass population, these results clearly assert an emerging need to rethink the planning and operation of the urban transit systems to balance the goal of efficiency maximization and the risk of infectious disease during extreme pandemics. 

\begin{figure}[ht!]
    \centering
    \subfloat[Disease dynamics of two scenarios in Network A\label{fig:sub_A}]{\includegraphics[width=0.6\linewidth]{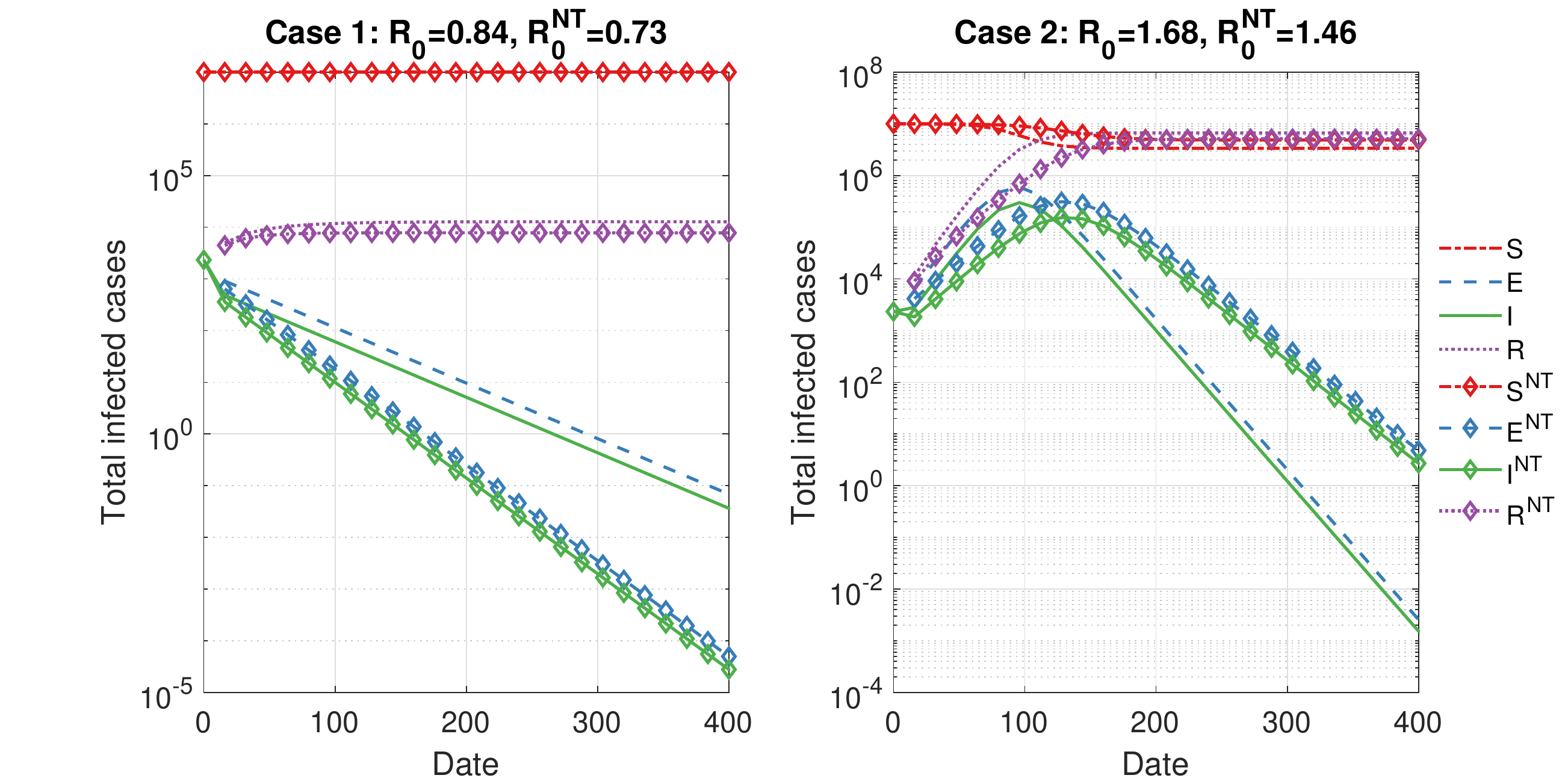}}\\
    \subfloat[Comparison of E and I curves for different cities using disease parameters in case 2 \label{fig:three_city}]{\includegraphics[width=0.55\linewidth]{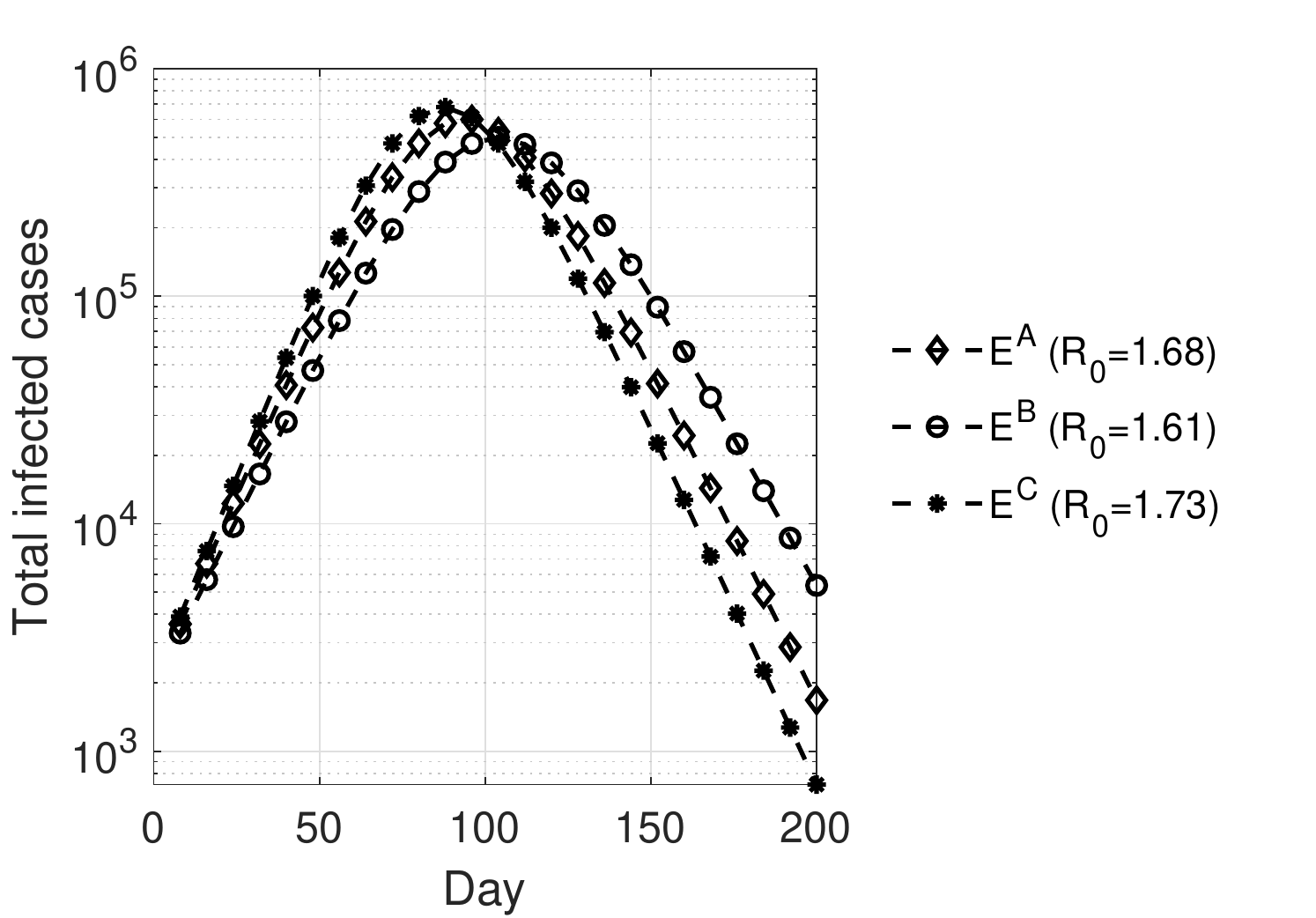}}
    \subfloat[Synchronization among different zones \label{fig:sync}]{\includegraphics[width=.4\linewidth]{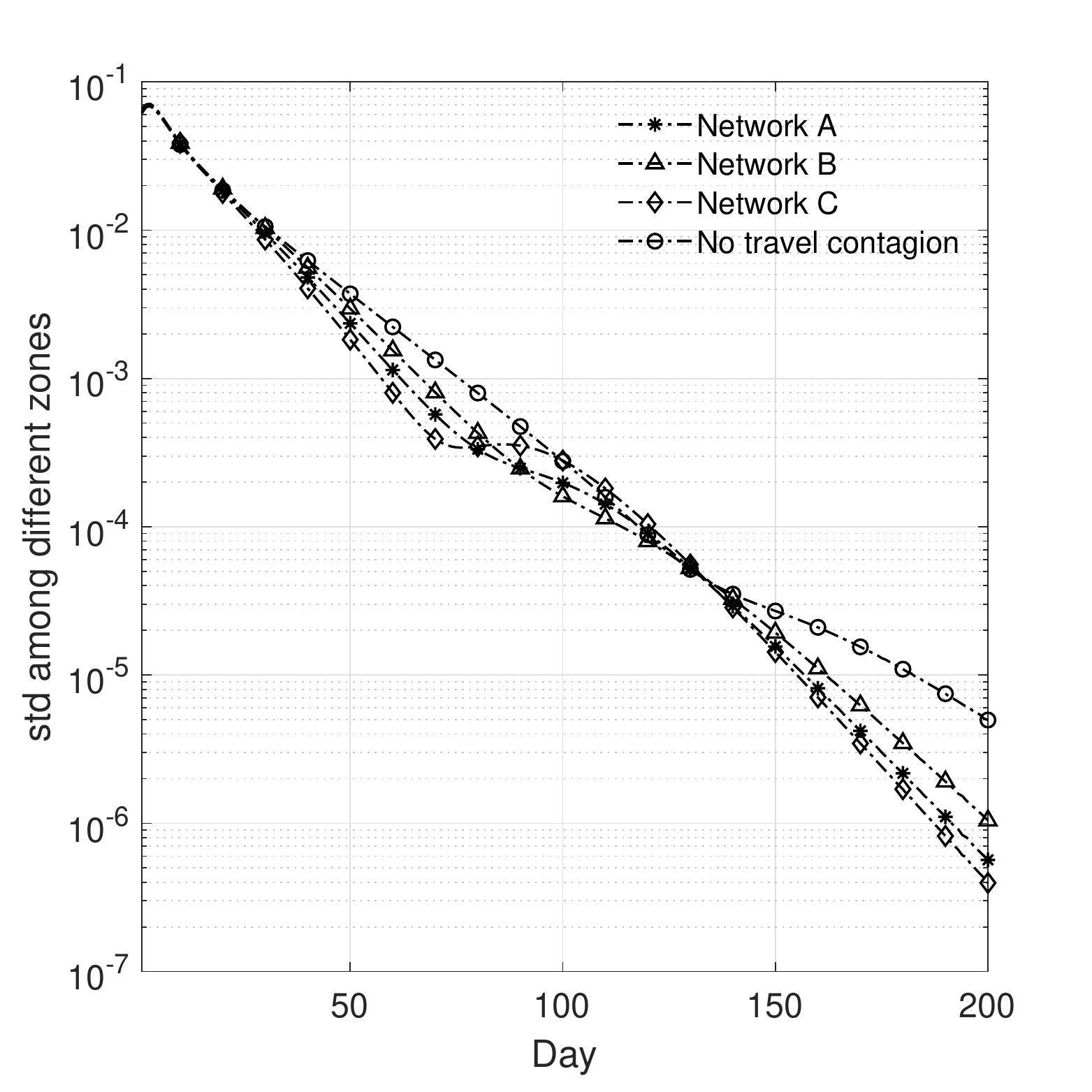}}
    \caption{Disease dynamics in two scenarios}
    \label{fig:sample_shape}
\end{figure}

\begin{figure}[ht!]
    \centering
    \includegraphics[width=\linewidth]{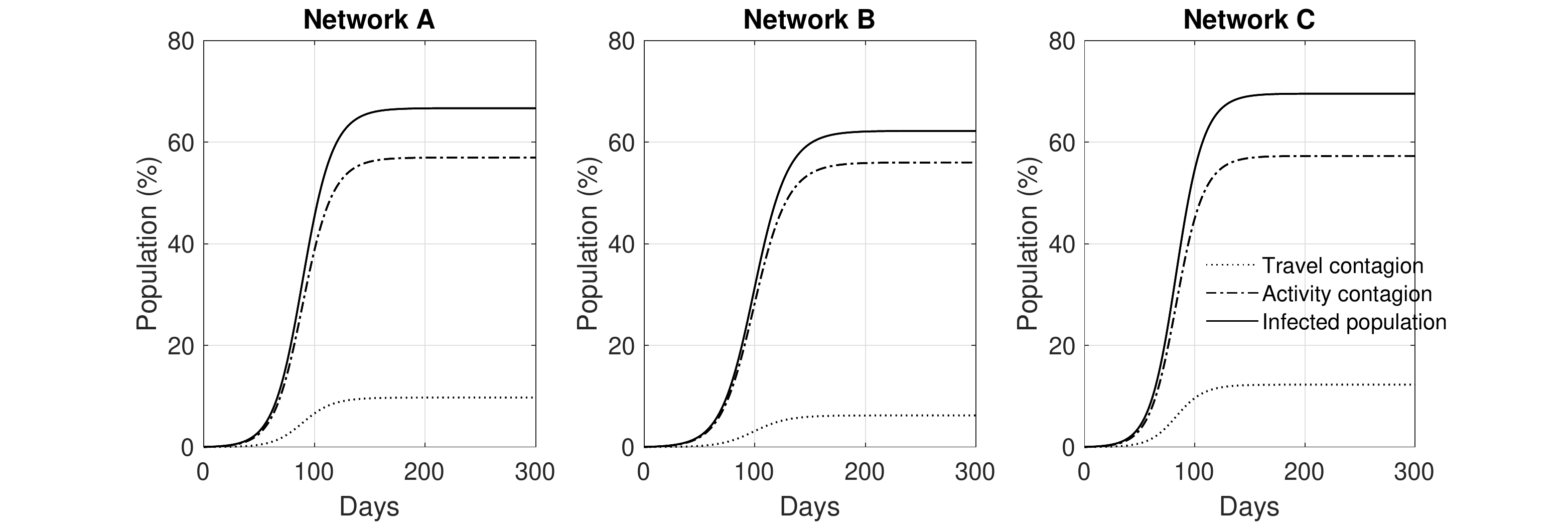}
    \caption{Cumulative dynamics of infected population, activity contagion and travel contagion in the three toy networks.}
    \label{fig:trans_dynamics}
\end{figure}

The previous discussion demonstrates the significant role played by the urban transportation system in boosting the spread of infectious diseases. Nevertheless, we further report that direct transmission during travel in the public transit systems may not constitute a major source of infections during an urban disease outbreak. To see this, we summary the cumulative dynamics of the total number of infectious (as the percentage of the total population), direct travel contagion and activity contagion, and the results can be found in Figure~\ref{fig:trans_dynamics}. It is observed that, depending on the particular network layout, direct travel contagion may account for a small proportion (10.0\% - 17.6\%) of total infected cases and the majority of the infections arise from non-travel related activity engagement such as work and entertainment. The reason is that the number of travelers in the transit system and the exposure duration are both lower than those during regular activities. While direct travel contagion is the most straightforward measure to quantify the impacts of the urban transportation system, the metric is not comprehensive as it does not capture the secondary cases produced by travel infections. To this end, we introduce the notion of induced travel contagion to measure the amount of infected population who are transmitted at activity locations by contacting people who are infected during travel. The induced travel contagion is calculated as the difference between the amount of activity contagion with and without travel infections, and the results are shown in Figure~\ref{fig:induced}. We note that the induced travel contagion in these three networks may account for an additional 4.27\% to 5.55\% of new infections. This implies that travel-related contagion may contribute to 16.9\% to 25.6\% of the total infections. Moreover, all the results are carried out based on the assumption that the transmission rate in the transit system is half of that at activity locations. We further illustrate how travel-related contagion may change by altering the travel transmission rate, and the results are shown in Figure~\ref{fig:sensitivity}. Specifically, we vary the travel transmission rate from 50\% to 250\% of the base travel transmission rate, and we observe that direct travel contagion increases linearly with respect to higher travel transmission rate, while the induced travel rate is found to be a concave function of the travel transmission rate. We further note that the linear relationship is only a local property within the evaluated range of travel transmission rate, which will turn into a concave function as we further increase the transmission rate. And the results suggest that, if we consider that the strengths of travel transmission and activity transmission being the same, the total travel-related infections may eventually account for over 35\% of the total infected cases during an urban disease outbreak. 

\begin{figure}[ht!]
    \centering
    \includegraphics[width=0.4\linewidth]{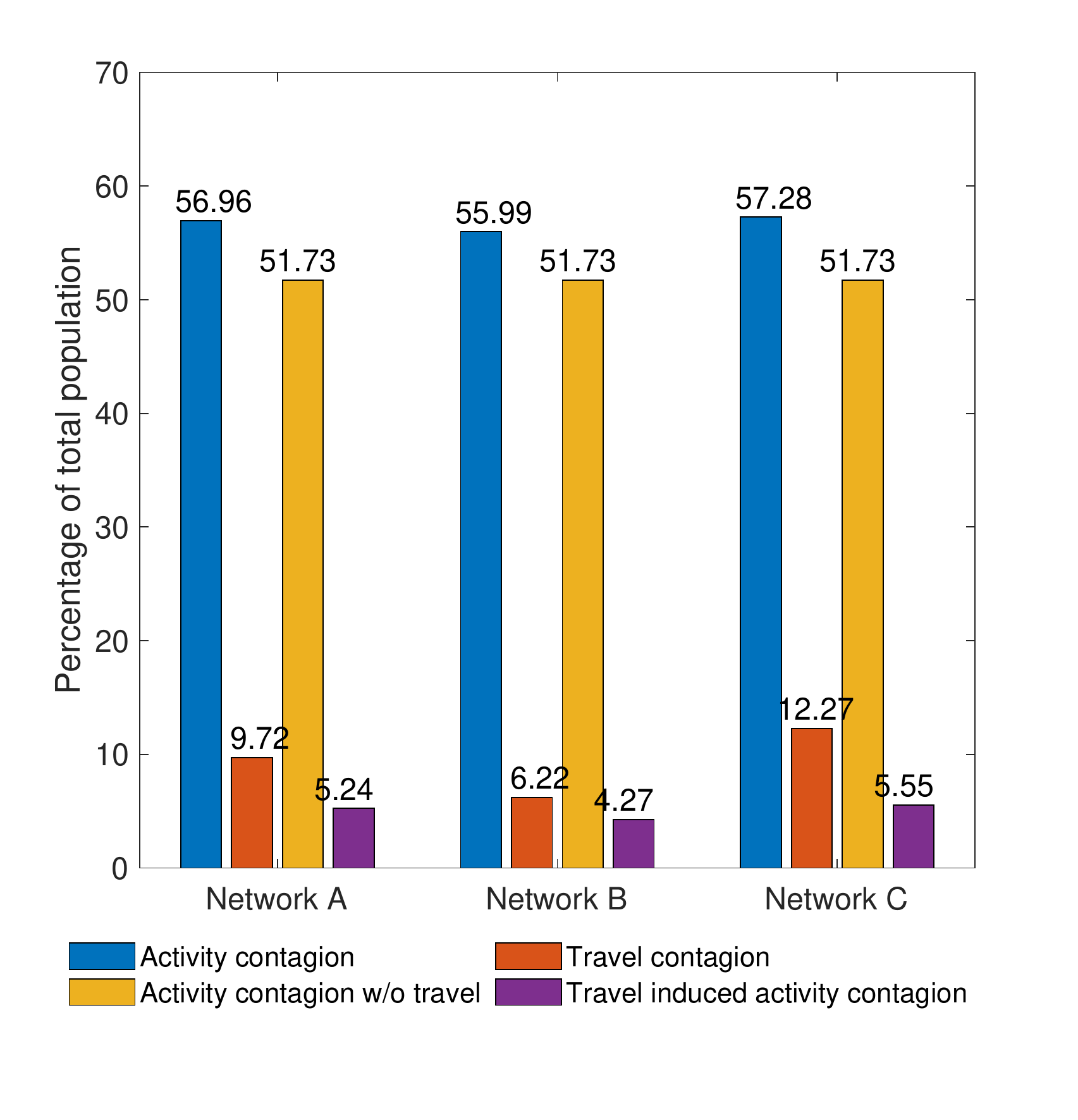}
    \caption{Number of induced travel contagions in different cities.}
    \label{fig:induced}
\end{figure}

\begin{figure}[ht!]
    \centering
    \includegraphics[width=\linewidth]{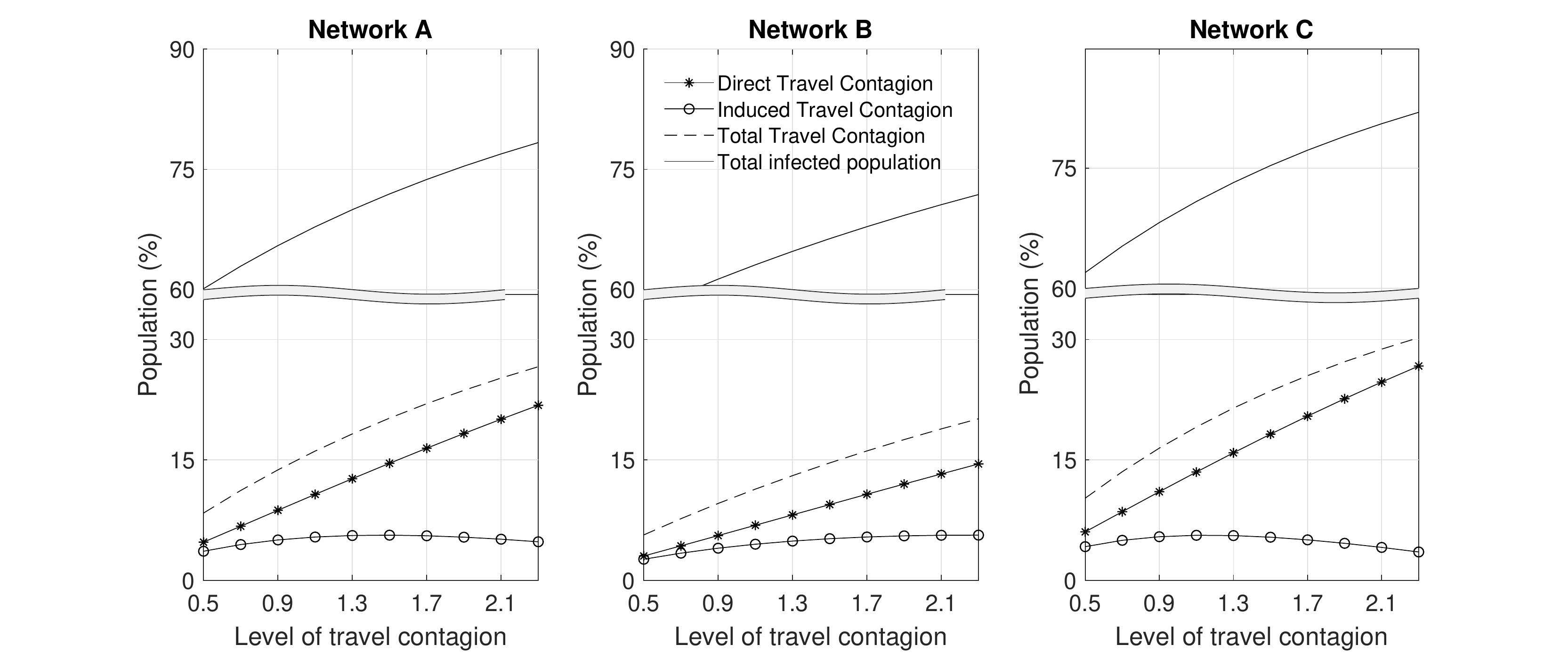}
    \caption{Sensitivity of travel contagion with respect to the change in transmission rate during travel ($\beta^T$)}
    \label{fig:sensitivity}
\end{figure}

\subsection{NYC experiments}
We use a combination of the United States 2010 census survey, 2010 NYC peripheral transportation study~\cite{nyc_periph_report2010}, NYC Neighborhood Tabulation Areas shapefile and the reported confirmed COVID-19 cases in NYC~\cite{nyc_covid_data} to prepare the data for our case study. To better account for the mobility dynamics, we set each step of the simulation as a 2-hour period and the modeling parameters are tailored to the time interval accordingly.  There are two sets of parameters to be calibrated. For the mobility dynamics, we use the 2010 census survey along with the 2010 peripheral transportation study, projected into 2020 population values, to obtain the population and the travel rate between nodes in the network. The mobility model at equilibrium estimates around 2.7 million daily commuters within NYC, and this number agrees with the value of the peripheral transportation study~\cite{nyc_periph_report2010}. For disease parameters, we set $\sigma$=$0.016$ (average duration of incubation period being 5.1 days~\cite{lauer2020incubation}) and $\gamma=0.013$ (average duration of recovery or death of 6.5 days) based on the calibrated parameters from the reported data in Hubei province, China~\cite{yang2020modified}. 

As the focus of our study is to understand the impacts of travel contagion rather than predicting the trajectories of the disease dynamics, we calibrate the disease transmission rate and the average number of contacts through a simple line search. And the parameters are finalized when the model is observed to align with the trend of the reported confirmed cases in general. The outputs from the calibrated model and its comparison with the reported confirmed cases in NYC can be seen in Figures~\ref{fig:nyc_all_trend} and~\ref{fig:nyc_borough}. In general, we find that the Trans-SEIR model may well capture the COVID-19 dynamics based on the calibrated parameters from the NYC census and travel data for daily commuting patterns. We observe that the dynamics of cumulative confirmed cases match well with the period between March 20 to April 5 when a large number of tests were performed so that the reported values were closer to the actual disease dynamics. Moreover, the model results are also consistent with the reported cases in each borough. This suggests the potential of the Trans-SEIR model for prediction tasks if more rigorous parameter calibration may be conducted. The disease dynamics is highly spatially heterogeneous. Following equation~\ref{eq:R0_value}, the model gives the $R_0$ value of 3.295 and the similar values were also reported in the literature~\cite{liu2020reproductive}. But this value only represents the average level of secondary infections over the entire NYC, and the values are significantly higher in certain boroughs. Specifically, we see that the transmission rates in Bronx, Brooklyn and Queens are over 30\% higher than that of Manhattan, despite Manhattan being the most densely populated borough. There are several possible reasons that may result in this observation. For instance, one possible explanation is the under-reporting or lack of testing in Manhattan. However, there may be other factors such as the differences in income level, level of education, age group, and the adoption of early preventative measures among the boroughs that lead to this high discrepancy. These other factors may worth further investigation.

\begin{figure}[ht!]
    \centering
    \includegraphics[width=0.6\linewidth]{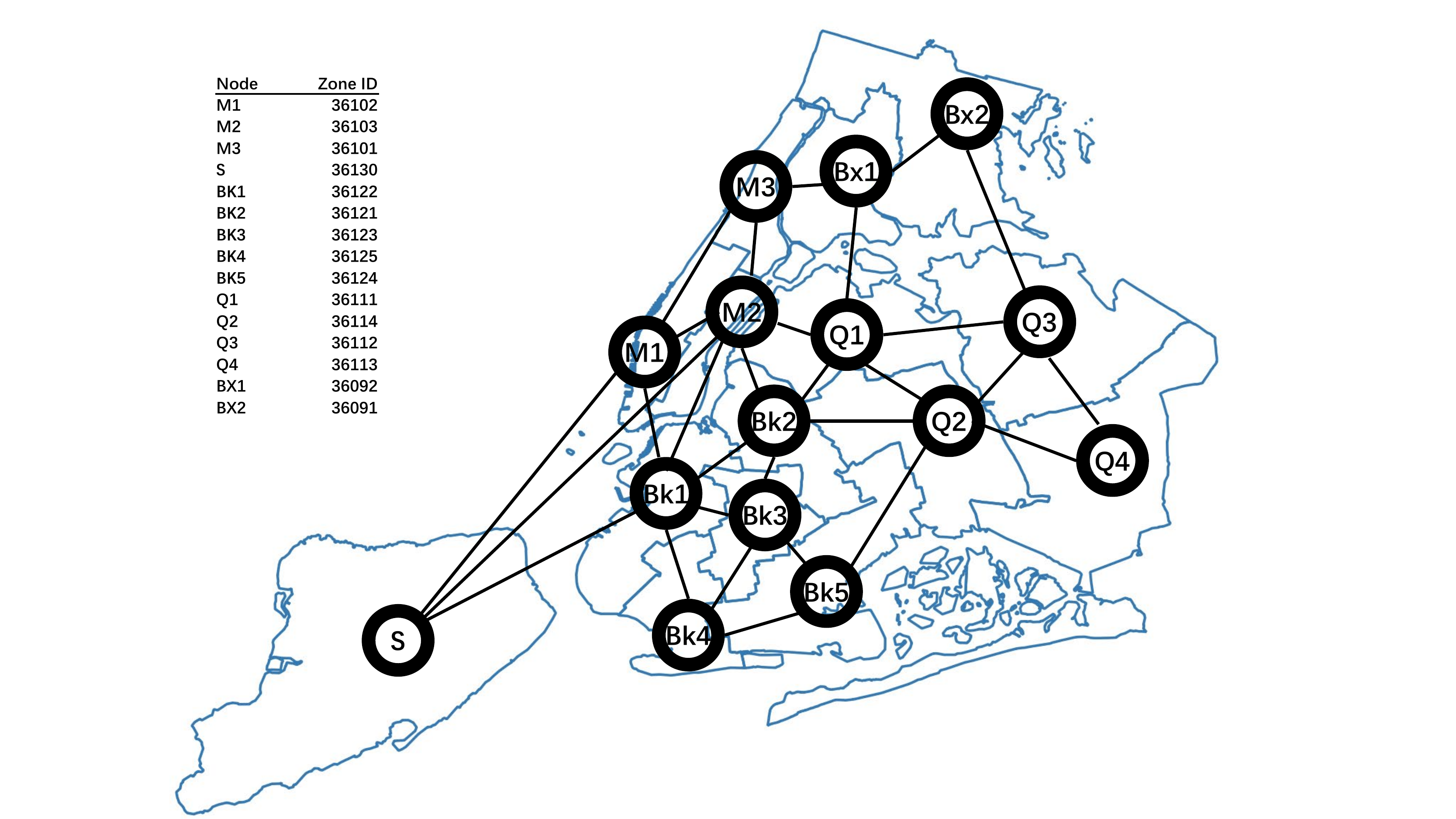}
    \caption{Network layout for the NYC case study}
    \label{fig:nyc_net}
\end{figure}

\begin{figure}[ht!]
    \centering
    \includegraphics[width=0.4\linewidth]{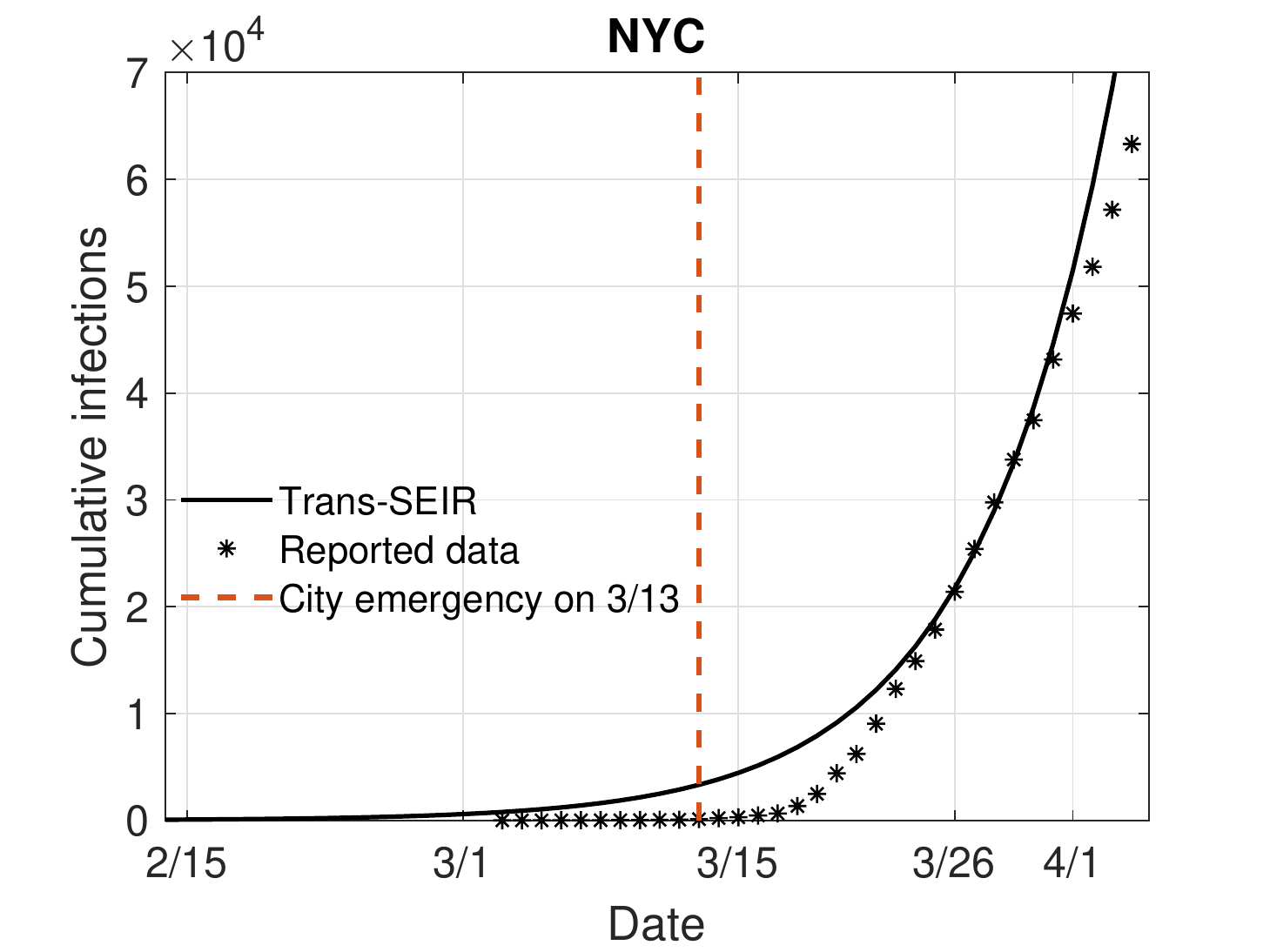}
    \caption{Fitted data for the Covid-19 case}
    \label{fig:nyc_all_trend}
\end{figure}

\begin{figure}[ht!]
    \centering
    \includegraphics[width=\linewidth]{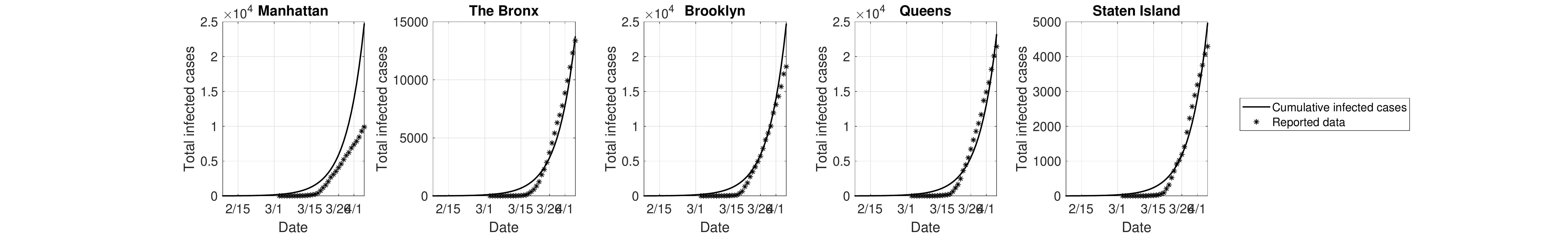}
    \caption{Fitted data for the Covid-19 case}
    \label{fig:nyc_borough}
\end{figure}

Since the Trans-SEIR model considers the spatial movement of urban travelers, the results also allow for revealing the travel segments and activity locations that are associated with the greatest number of infections and we visualize the results in Figure~\ref{fig:travel_flow}. The strong spatial heterogeneity in terms of both activity and travel contagion can be readily observed in Figure~\ref{fig:nyc_hist}, where total contagion during travel may account for 28.8\% of the total cases in NYC as of March 26 but the values differ based on the residential locations of the commuters. In particular, Staten Island is the place with the minimum travel contagion as most of the trips are within the same area and travelers are least dependent on public transit. On the contrary, the north Manhattan area has the highest ratio of travel contagion with its heavy usage of buses and metros. For commuters of the same residential location, the split ratios of travel contagion and activity contagion also vary significant based on their activity locations. Figure~\ref{fig:nyc_travel} presents the relative proportion of travelers (represented as the thickness of the flow curves) who are residents of zones in the top and get infected in the transit system while traveling to and from the bottom zones. It can be observed that east Manhattan is the activity location that is associated with the highest number of the infected population during travel. Moreover, east Manhattan is also the primary destination that contributes to the highest number of travel contagions for residents of each zone except for Staten Island where most of the commuting trips take place within the same area. The major reasons behind this observation are threefold: east Manhattan is the major trip attractor for other areas in NYC, transit such as bus and metro are the major modes for commuting to the area, and long commuting trips are expected for people living outside the Manhattan area. These lead to both high population density and prolonged contact duration for the transit routes linking to east Manhattan which significantly exacerbates the contagion risk level. Besides, we find that west \& lower Manhattan and Ft. Greene/Bay of Brooklyn are the other two major destination locations for infections during travel. The source of travel contagions who arrive at west \& lower Manhattan are again the commuters from outside Manhattan. But the corresponding source of travelers who arrive at Ft. Greene/Bay of Brooklyn are primarily the residents within the Brooklyn area, and similar patterns are observed for the other activity locations. The distribution of activity contagion largely resembles that of travel contagion as shown in Figure~\ref{fig:nyc_act}. The major difference is that activity locations are mostly observed within the same location for areas such as Ft.Greene of Brooklyn, northwest Queens and the Bronx. The reason is that majority of the residents in these locations also have activity locations in the same areas, and trips within these areas are mostly made by private vehicles so that travel contagion is not the major threat. Based on these findings, we arrive at a better understanding of the spatial trajectories of the disease dynamics where the east and west \& lower Manhattan areas are likely the intermediate stops before the disease reaches every corner of NYC. This also hints at the potential locations where entrance screening may be placed to mitigate the spread of the diseases, which require further validation based on the results of the optimal control problem. 

\begin{figure}[htbp!]
    \centering
    \subfloat[Number of travel contagion and activity contagion in each area\label{fig:nyc_hist}]{\includegraphics[width=0.8\linewidth]{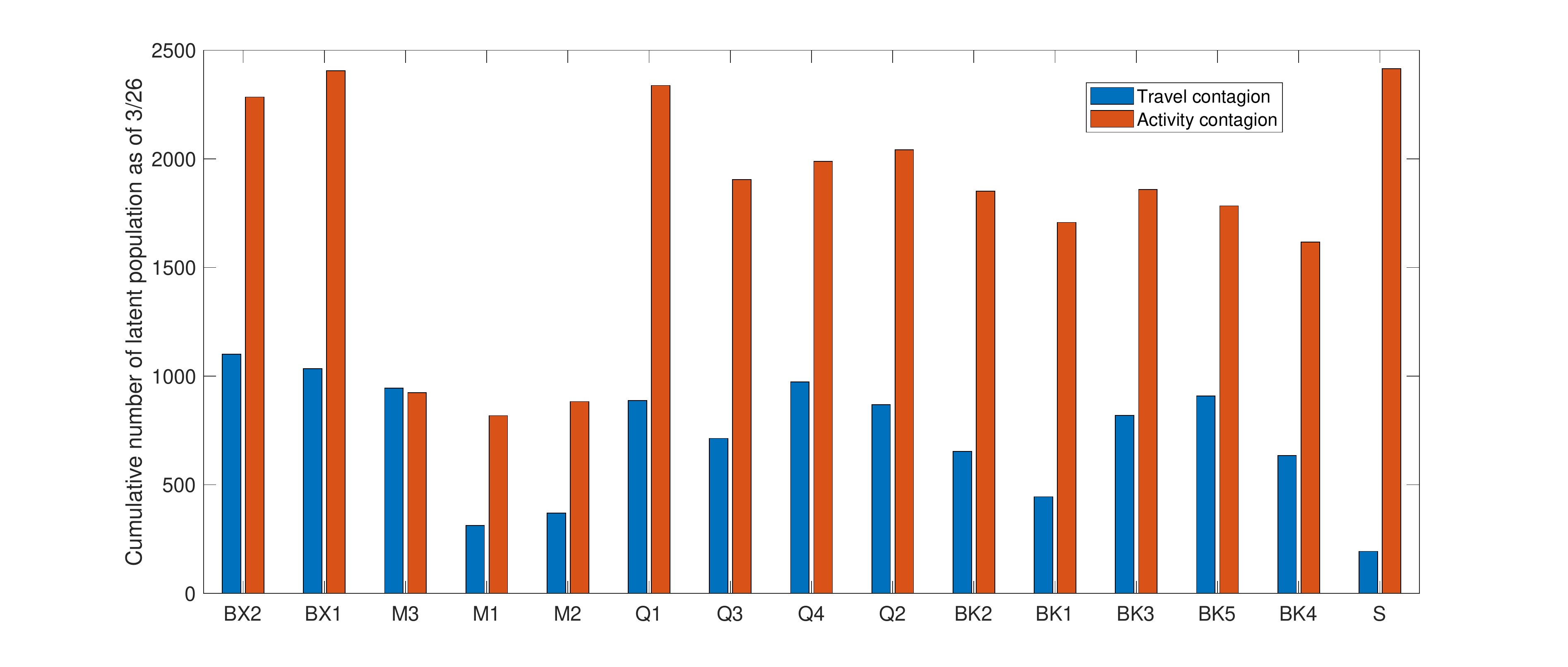}}\\
    \subfloat[Alluvial flow diagram for the amount of travel contagion between pairs of zones.\label{fig:nyc_travel}]{\includegraphics[angle=-90,width=0.8\linewidth]{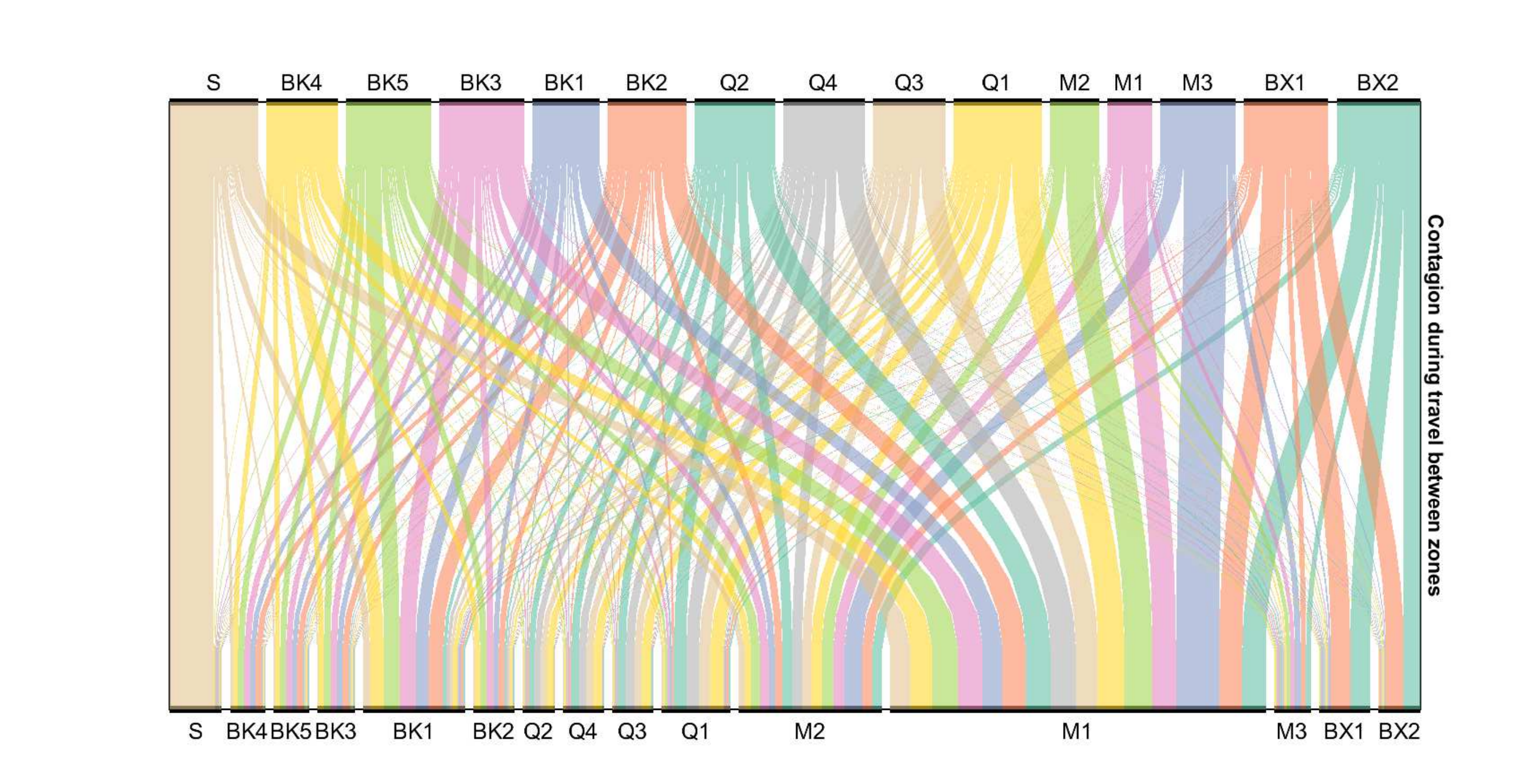}}\\
    \subfloat[Alluvial flow diagram for the amount of activity contagion between pairs of zones.\label{fig:nyc_act}]{\includegraphics[angle=-90,width=0.8\linewidth]{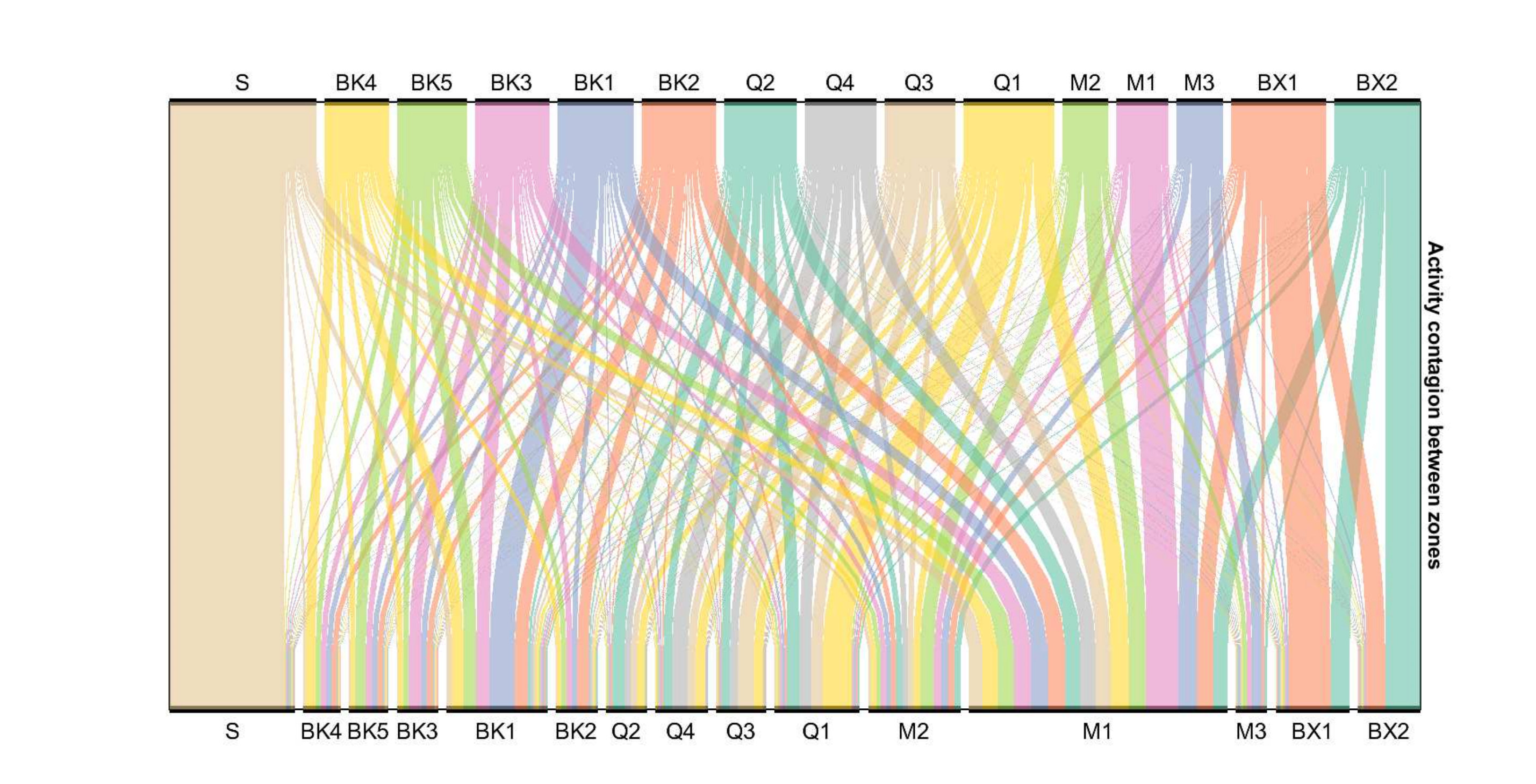}}
    \caption{The distributions of travel contagion and activity contagion during the COVID-19 outbreak in NYC, measured by the amount of people in latent class (E) as of March 26. }
    \label{fig:travel_flow}
\end{figure}

\begin{figure}[ht!]
    \centering
    \includegraphics[width=0.4\linewidth]{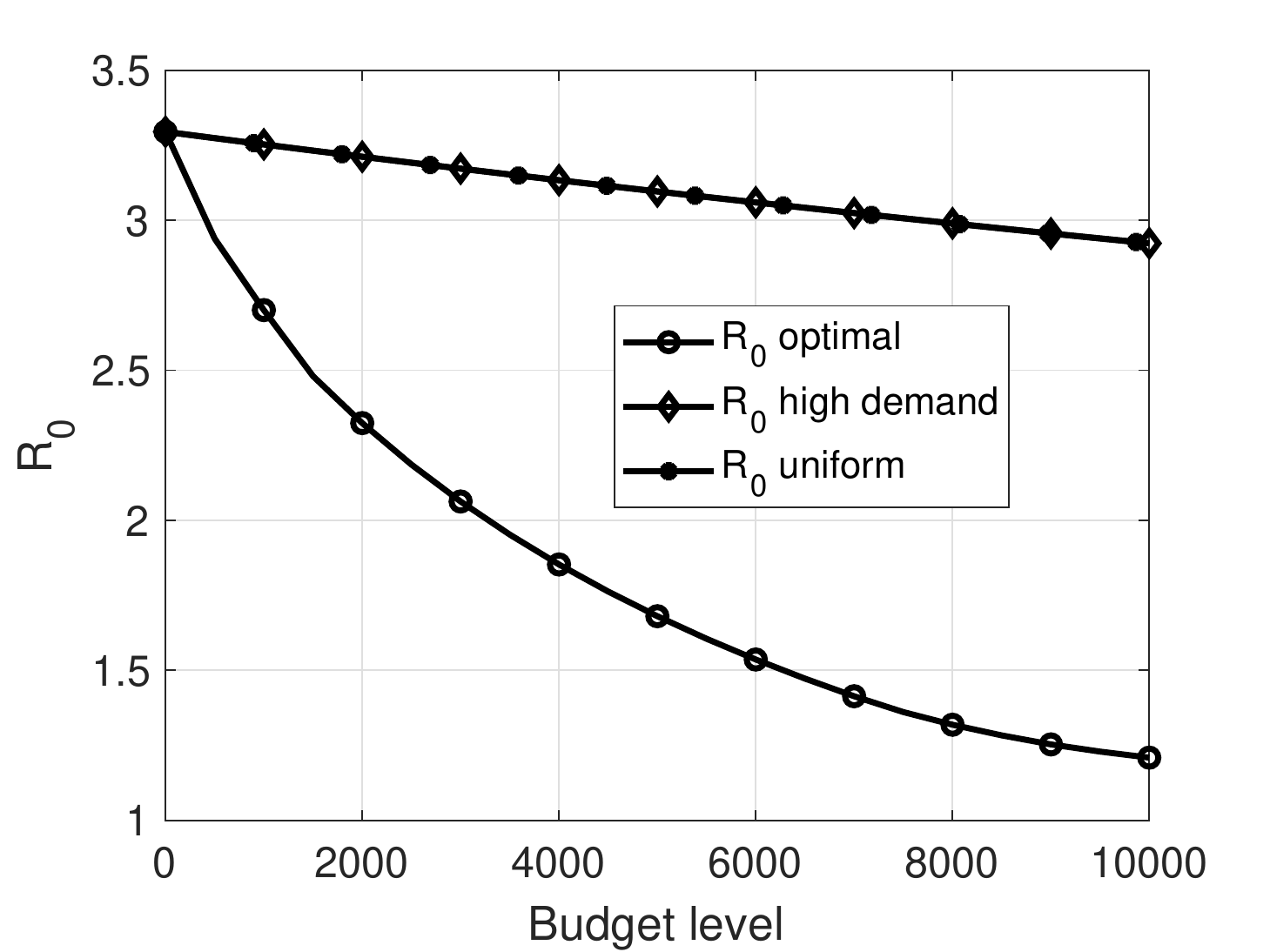}
    \caption{Effectiveness of transportation control with different budget level}
    \label{fig:nyc_control}
\end{figure}

\begin{figure}[htbp!]
    \centering
    \includegraphics[width=0.8\linewidth]{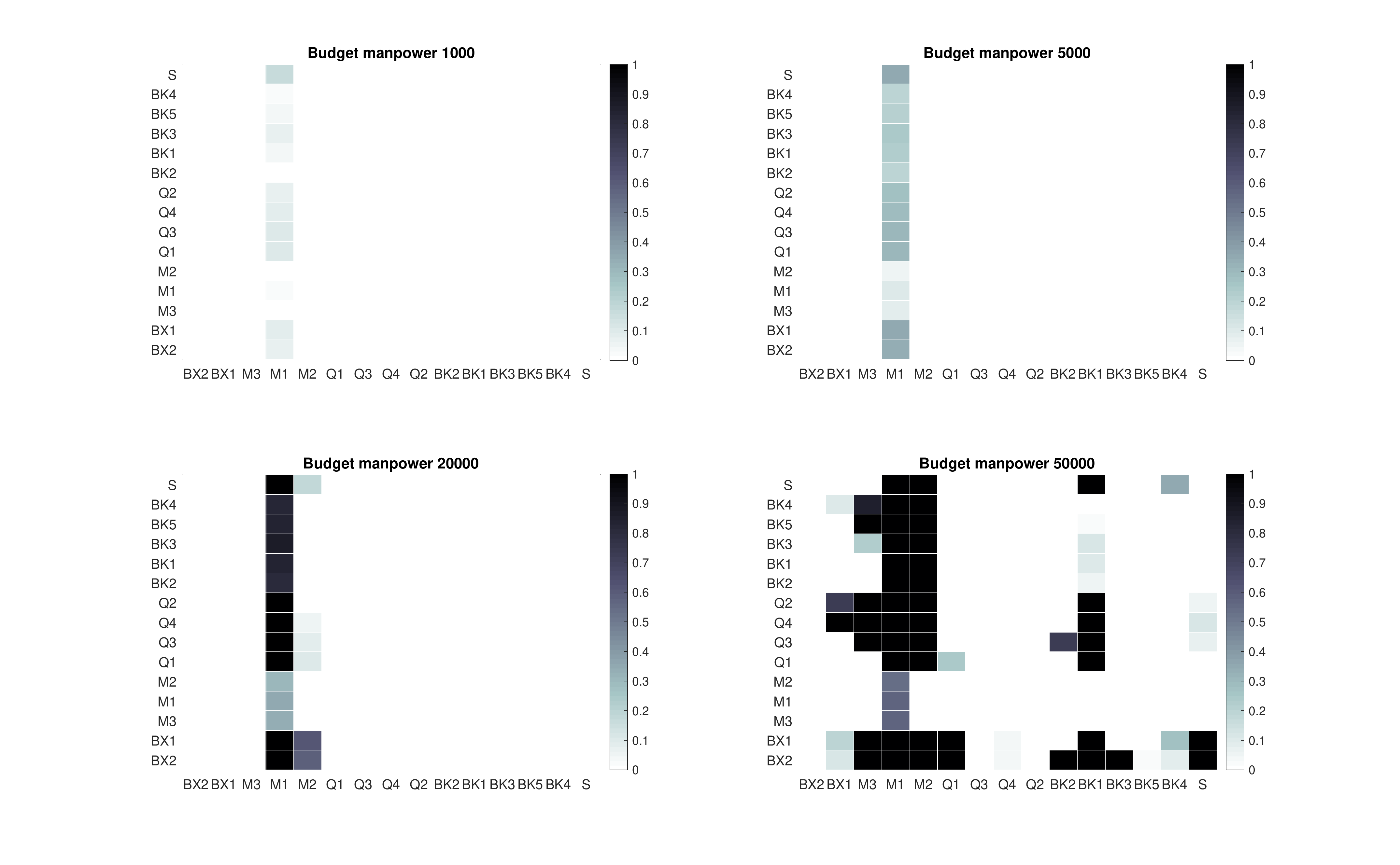}
    \caption{Distribution of control resources under different budget level}
    \label{fig:nyc_budget_allocation}
\end{figure}

We summarize the results of optimal entrance control in Figures~\ref{fig:nyc_control} and~\ref{fig:nyc_budget_allocation}. We report that entrance control for urban transit system can be a highly effective countermeasure against the outbreak of infectious diseases in urban areas, for both preventative and reactive purposes. If limited resources are distributed in an optimal manner, we are able to curb the disease and drive it to the disease-free equilibrium state, and will contribute to a significant reduction in the number of total infections. To validate the proposed solution approach, we compare the effectiveness among optimal control as devised in our study and two other benchmark methods: the uniform control where equal resources are distributed among all travel segments and the high-demand control where travel segments with higher demand are prioritized for resources. The measure of effectiveness is the $R_0$ value after control as shown in Figure~\ref{fig:nyc_control}, where the x-axis denotes the number of resources spent (manpower used for screening travelers). We observe that, by only focusing on medium and high capacity modes and if we can distribute the resources optimally, we are able to lower the $R_0$ to 1.2 by screening approximately 10,000 daily travelers per time step, or less than 4.5\% of total daily commuters. While the screening of 10,000 travelers can be expensive in practice and may lead to excessive externalities such as heavy delay, we also find that an exponential rate of reduction in $R_0$ can be achieved with optimal control strategy for a small number of initial efforts. If we may perform the screening of 2000 travelers with a success rate of 70\%, we will be able to lower the $R_0$ in NYC to 2.32 which represents a 29.6\% reduction in the number of secondary infections and implies significant potential saving of lives and medical resources.  Note that this level of control effectiveness is achieved through the compartment model and control strategy being implemented at the aggregate level, an encouraging finding is that the outcomes of the optimal control approach are comparable to that of the target control strategy at the individual level with perfect information on the contact network structure and the state of each travelers~\cite{qian2020scaling}.  Meanwhile, the other two benchmark approaches are found to be barely effective. Despite direct and induced travel contagions, we need to pay attention to the case that infectious people will need to travel using the urban transit system to further spread the disease. If the resources and screening manpower can be optimally allocated, we can effectively identify these infectious travelers, put them into quarantine and stop them from producing new infections. To this end, we also visualize the optimal distribution of resources across the travel segments as shown in Figure~\ref{fig:nyc_budget_allocation}. The major findings echo to the general observations in Figure~\ref{fig:nyc_all_trend}, where east Manhattan is the target area to place the majority of the resources and should be exclusively controlled with very limited resources (budget<5000). Moreover, for commuters with east Manhattan as the activity location, the optimal solution suggests targeting those from Staten Island, Queens and the Bronx with top priority. This can help to prevent infectious travelers from resulting in secondary cases within their areas where massive activity contagions are observed. With more resources available, the optimal solution recommends focusing on residents of the above-mentioned areas who have their activity locations in west \& and lower Manhattan and part of the Brooklyn area, before placing stricter screening for residents within the Manhattan area due to the relatively low activity contagion rates. In conclusion, the results of the optimal control method suggest a highly promising direction for the control of urban transportation systems in densely populated areas. And the Trans-SEIR model with the travel control approach may have significant implications towards the operation of public transit systems after the reopening of businesses to mitigate further risks of the COVID-19 outbreak.

\section{Conclusion}

In this study, a realistic Trans-SEIR model is presented for understanding the spread of infectious disease in urban areas considering the spatial-varying mobility dynamics and the presence of both travel contagions and activity contagions. The Trans-SEIR model starts with the dynamic system to capture the population movement among different zones and derive the equilibrium flow pattern for the urban areas. Based on the stable commuting pattern, the spread of infectious disease in urban areas is formulated as a spatial SEIR system with travel contagion, and an entrance control framework for urban transportation is further proposed for the optimal allocation of limited resources for reducing the risks from the infectious diseases. The presented modeling approach along with the optimal control framework will allow for advances in understanding the disease dynamics in urban areas and contribute to the framing controlling strategies and policies for mitigating the risk of infectious diseases from the urban transportation perspective. 

\bibliographystyle{model1-num-names}
\bibliography{references}
\end{document}